\newif\ifStable\Stabletrue

\documentclass[11pt,letterpaper]{article}
\usepackage[margin=1.25in,dvips]{geometry}
\usepackage{amsthm}
\usepackage{fullpage}

\usepackage{times}
\usepackage{graphicx}
\usepackage{amssymb}
\usepackage{amsmath}
\usepackage{float}
\usepackage{url}
\usepackage{enumerate}
\usepackage{subfig}
\usepackage[boxed,section]{algorithm}
\usepackage{algpseudocode}
\usepackage{color}
\usepackage[table]{xcolor}

\newtheorem{theorem}{Theorem}[section]
\newtheorem{lemma}[theorem]{Lemma}

\newtheorem{definition}[theorem]{Definition}
\newtheorem{claim}[theorem]{Claim}

\newcommand{\sabs}[1]{|#1|}
\newcommand{\abs}[1]{\left|#1\right|}
\newcommand{\floor}[1]{\left\lfloor#1\right\rfloor}

\newcommand{\norm}[2]{\lVert#2\rVert_{#1}}

\newcommand{\wh}{\widehat}

\newcommand{\hide}[1]{}

{\makeatletter
 \gdef\xxxmark{%
   \expandafter\ifx\csname @mpargs\endcsname\relax 
     \expandafter\ifx\csname @captype\endcsname\relax 
       \marginpar{xxx}
     \else
       xxx 
     \fi
   \else
     xxx 
   \fi}
 \gdef\xxx{\@ifnextchar[\xxx@lab\xxx@nolab}
 \long\gdef\xxx@lab[#1]#2{{\bf [\xxxmark #2 ---{\sc #1}]}}
 \long\gdef\xxx@nolab#1{{\bf [\xxxmark #1]}}
\long\gdef\xxx@lab[#1]#2{}\long\gdef\xxx@nolab#1{}%
}

\DeclareMathOperator*{\median}{median}
\DeclareMathOperator*{\mean}{mean}
\DeclareMathOperator{\supp}{supp}

\newcommand{\andop}{\text{ and }}
\DeclareMathOperator{\round}{round}

\DeclareMathOperator*{\E}{\mathbb{E}}
\DeclareMathOperator{\symdiff}{\triangle}
\newcommand\R{\mathbb{R}}
\newcommand{\C}{{\mathbb C}}

\newcommand\eps{\epsilon}

\newcommand{\xref}[1]{\S\ref{#1}}

\newcommand{\textred}[1]{\textcolor{red}{#1}}
\newenvironment{Itemize}%
{\begin{itemize}%
\setlength{\itemsep}{0pt}%
\setlength{\topsep}{0pt}%
\setlength{\partopsep}{0pt}%
\setlength{\parskip}{0pt}}%
{\end{itemize}}
\newenvironment{Enumerate}%
{\begin{enumerate}%
\setlength{\itemsep}{0pt}%
\setlength{\topsep}{0pt}%
\setlength{\partopsep}{0pt}%
\setlength{\parskip}{0pt}}%
{\end{enumerate}}
\newfont{\aufnt}{phvr8t at 12pt}
\newfont{\affaddr}{phvr8t at 10pt}

\begin{document}

\title{Sample-Optimal Average-Case Sparse Fourier Transform in Two Dimensions}

\author{Badih Ghazi \and Haitham Hassanieh \and Piotr Indyk \and Dina Katabi \and Eric Price \and Lixin Shi \\ Massachusetts Institute of Technology}
\date{{\tt \{badih,haithamh,indyk,dk,ecprice,lixshi\}@mit.edu}}

\maketitle
\begin{abstract}
  We present the first sample-optimal sublinear time algorithms for
  the sparse Discrete Fourier Transform over a two-dimensional
  $\sqrt{n} \times \sqrt{n}$ grid.  Our algorithms are analyzed for
  \emph{average case} signals.  For signals whose spectrum is
  exactly sparse, our algorithms use $O(k)$ samples and run in $O(k
  \log k)$ time, where $k$ is the expected sparsity of the signal.
  For signals whose spectrum is approximately sparse,
  our algorithm uses $O(k \log n)$ samples and runs in $O(k \log^2 n)$
  time; the latter algorithm works for $k=\Theta(\sqrt{n})$. The
  number of samples used by our algorithms matches the known lower
  bounds for the respective signal models.

  By a known reduction, our algorithms give similar results for the
  one-dimensional sparse Discrete Fourier Transform when $n$ is a
  power of a small composite number (e.g., $n = 6^t$).
\end{abstract}

\section{Introduction}\label{sec:intro}

The Discrete Fourier Transform (DFT) is a powerful tool used in many
 domains. Multimedia data sets, including video and
images, are typically processed in the frequency domain to compress
the data~\cite{jpeg,mpeg2,ImageVideoCompression}.
 Medicine and
biology rely on the Fourier transform to analyze the output of a
variety of tests and experiments including MRI~\cite{MRINishimura},
NMR~\cite{nmrSIFT} and ultrasound
imaging~\cite{Ultrasoundbook}. Other applications include astronomy and radar systems.

The fastest known algorithm for computing the DFT is the Fast Fourier
Transform (FFT). It computes the DFT of a signal with size $n$
in $O(n \log n)$ time. Although it is not known whether this algorithm is optimal, any general algorithm for computing the exact DFT must take time at least proportional to its output size,
i.e., $\Omega(n)$. In many applications, however, most of the Fourier coefficients of a
signal are small or equal to zero, i.e., the output of the DFT is
(approximately) {\em sparse}. This sparsity provides the rationale underlying compression
schemes for image and video signals such as JPEG and MPEG. In fact, all of the aforementioned applications involve sparse data. 

For sparse signals, the $\Omega(n)$ lower bound for the complexity of
DFT no longer applies.  If a signal has a small number $k$ of nonzero
Fourier coefficients---the {\em exactly $k$-sparse} case---the output
of the Fourier transform can be represented succinctly using only $k$
coefficients. Hence, for such signals, one may hope for a DFT
algorithm whose runtime is sublinear in the signal size $n$. Even
in the more general {\em approximately $k$-sparse} case,
it is possible in principle to find the large components of its Fourier transform in sublinear time. 

The past two decades have witnessed significant advances in sublinear
sparse Fourier algorithms.  The first such algorithm (for the Hadamard
transform) appeared in~\cite{KM} (building on~\cite{GL}). Since then,
several sublinear sparse Fourier algorithms for complex inputs have been
discovered \cite{Man,GGIMS,AGS,GMS,Iw,Ak,HIKP,HIKP2,LWC,BCGLS,HAKI}.
The most efficient of those algorithms\footnote{See the discussion in the Related Work section.}, given in~\cite{HIKP2}, offers the following performance guarantees:
   \begin{Itemize}
   \item For signals that are exactly $k$-sparse, the algorithm runs in $O(k \log n)$ time.
   \item For the approximately sparse signals, the algorithm runs in $O(k \log n \log(n/k))$ time.
 \end{Itemize}

 Although the aforementioned algorithms are very efficient, they
 nevertheless suffer from limitations.  Perhaps the main limitation is
 that their sample complexity bounds are equal to the their running
 times.  In particular, the sample complexity of the first algorithm
 (for the exactly $k$-sparse case) is $\Theta(k \log n)$, while the sample
 complexity of the second algorithm (approximately sparse) is $\Theta(k \log(n) \log(n/k))$.  The
 first bound is suboptimal by a logarithmic factor, as it is known
 that one can recover any signal with $k$ nonzero Fourier
 coefficients from $O(k)$ samples~\cite{reed-solomon}, albeit in
 super-linear time. The second bound is a logarithmic factor away from
 the lower bound of $\Omega(k \log(n/k))$~\cite{PW} established for
 non-adaptive algorithms\footnote{An algorithm is {\em adaptive} if it
   selects the samples based on the values of the previously sampled
   coordinates. If the positions of the samples are chosen in advance
   of the sampling process, the algorithm is called {\em
     non-adaptive}. All algorithms given in this paper are
   non-adaptive.}; a slightly weaker lower bound of $\Omega(k
 \log(n/k) /\log \log n)$ applies to adaptive algorithms as
 well \cite{HIKP2}. In most applications, low sample complexity is at
 least as important as efficient running time, as it implies reduced
 signal acquisition or communication cost.

 Another limitation of the prior algorithms is that most of them are
 designed for one-dimensional signals. This is unfortunate, since
 multi-dimensional instances of DFT are often particularly
 sparse. This situation is somewhat alleviated by the fact that the
 two- dimensional DFT over $p \times q$ grids can be reduced to the
 one- dimensional DFT over a signal of length
 $pq$~\cite{GMS,Iw-arxiv}. However, the reduction applies only if $p$
 and $q$ are relatively prime, which excludes the most typical case of $m
 \times m$ grids where $m$ is a power of $2$. The only prior algorithm
 that applies to general $m\times m$ grids, due to~\cite{GMS}, has
 $O(k \log^c n)$ sample and time complexity for a rather large value of
 $c$.  If $n$ is a power of $2$, a two-dimensional adaptation of
 the~\cite{HIKP} algorithm (outlined in the appendix) has roughly $O(k
 \log^3 n)$ time and sample complexity.

\paragraph{Our results} In this paper, we present the first sample-optimal sublinear time algorithms for the Discrete Fourier Transform over a two- dimensional $\sqrt{n} \times \sqrt{n}$ grid.
Unlike the aforementioned results, our algorithms 
are analyzed in the {\em average case}. Our input distributions are natural. For the exactly sparse case, we assume the Bernoulli model: each spectrum coordinate is nonzero with probability $k/n$, in which case the entry assumes an arbitrary value predetermined for that position\footnote{Note that this model subsumes the scenario where the values of the nonzero coordinates are chosen i.i.d. from some distribution.}.  For the approximately sparse case, we assume that the spectrum $\wh{x}$ of the signal is a sum of two vectors: the signal vector, chosen from the Bernoulli distribution, and the noise vector, chosen from the Gaussian distribution (see Section~\xref{sec:defs} Preliminaries for the complete definition). 
These or similar\footnote{A popular alternative is to use the hypergeometric distribution over the set of nonzero entries instead of the Bernoulli distribution. 
The advantage of the former is that it yields vectors of sparsity {\em exactly} equal to $k$.
In this paper we opted for the Bernoulli model since it is simpler to analyze. However, both models are quite similar. In particular,  for large enough $k$, the actual sparsity of vectors in the Bernoulli model is sharply concentrated around $k$.} distributions
 are often used as test cases for empirical evaluations of sparse Fourier Transform algorithms \cite{IGS,HIKP,LWC} or theoretical analysis of their performance~\cite{LWC}.
 
 The algorithms succeed with a constant probability. The notion of success depends on the scenario considered. For the exactly sparse case, an algorithm is successful if it recovers the spectrum exactly. 
 For the approximately sparse case, the algorithm is successful if it  reports a signal with spectrum $\wh{z}$ such
  that 
 \begin{align}
 \label{eq:guarantee}
 \norm{2}{\wh{z}-\wh{x}}^2 = O( \sigma^2 n) +\norm{2}{\wh{x}}^2/n^c
 \end{align}
 where $\sigma^2$ denotes the variance of the normal distributions
 defining each coordinate of the noise vector, and where $c$ is any constant. Note that any
 $k$-sparse approximation to $\wh{x}$ has error $\Omega(\sigma^2 n)$ with overwhelming probability, and that the second term in the bound  in Equation~\ref{eq:guarantee} is subsumed by the first term as long as the signal-to-noise ratio is at most polynomial, i.e., $\norm{2}{\wh{x}} \le n^{O(1)} \sigma$.
 See Section~\xref{sec:defs} for further discussion. 

The running time and sample complexity bounds are depicted in the following table. 
We assume that $\sqrt{n}$ is a power of $2$.

\begin{center}
\begin{tabular}{|c|c|c|c|c|}
\hline
Input & Samples & Time & Assumptions\\
\hline
Sparse & $k$ & $k \log k$ & $k=O(\sqrt{n})$ \\
Sparse & $k$ & $k \log k $ & \\
& &  $+ k (\log\log n)^{O(1)} $ & \\
Approx. sparse & $k \log n$ & $k \log^2 n$ & $k=\Theta(\sqrt{n})$ \\
\hline
\end{tabular}
\end{center}

The key feature of our algorithms is that their sample complexity bounds
are optimal, at least in the non-adaptive case.  For the exactly
sparse case, the lower bound of $\Omega(k)$ is immediate.  For the
approximately sparse case, we note that the $\Omega(k \log(n/k))$
lower bound of~\cite{PW} holds even if the spectrum is the sum of a
$k$-sparse signal vector in $\{0, 1, -1\}^n$ and Gaussian noise. The
latter is essentially a special case of the distributions handled by
our algorithm, and we give a full reduction in
Appendix~\ref{app:lower}.  From the running time perspective, our
algorithms are slightly faster than those in \cite{HIKP2}, with the
improvement occurring for low values of $k$.

An additional feature of the first algorithm is its simplicity and therefore its low ``big-Oh'' overhead. Our preliminary experiments on random sparse data indicate that the algorithm for exactly sparse case yields substantial improvement over 2D FFTW, a highly efficient implementation of 2D FFT. In particular, for $n=2^{22}$ (a $2048\times 2048$ signal) and $k=1024$, the algorithm is 100$\times$ faster than 2D FFTW. To the best of our knowledge, this is the first implementation of a 2D sparse FFT algorithm. For the same $n$ and $k$, the algorithm has a comparable running time (1.5$\times$ faster) to the 1D exactly sparse FFT in~\cite{HIKP2} while using 8$\times$ fewer samples. We expect that the algorithm or its variant will be efficient on non-random data as well, since the algorithm can randomize the positions of the coefficients using random two-dimensional affine transformations (cf. Appendix~\ref{app:inefficient_2D}).  Even though the resulting distribution is not fully random, it has been observed that random affine transformations work surprisingly well on real data~\cite{MV}.

\paragraph{Our techniques} 
Our first algorithm for $k$-sparse signals is based on the following idea. Recall that one way to compute the two-dimensional DFT of a signal $x$ is to apply the one-dimensional DFT to each column and then to each row. Suppose that $k=a \sqrt{n}$ for $a<1$. In this case, the expected number of nonzero entries in each row is less than $1$. If {\em every} row contained exactly one nonzero entry, then the  DFT could be computed via the following two step process. In the first step, we select  the first two columns of $x$, denoted by $u^{(0)}$ and $u^{(1)}$, and compute their DFTs $\wh{u}^{(0)}$ and $\wh{u}^{(1)}$. 
Let $j_i$ be the index of the unique nonzero entry in the $i$-th row of $\wh{x}$, and let $a$ be its value.  Observe that
 $\wh{u}^{(0)}_{i} = a$ and $\wh{u}^{(1)}_{i} = a \omega^{-j_i}$ (where $\omega$ is a primitive $\sqrt{n}$-th root of unity), as these are the first two entries of the inverse Fourier transform of a $1$-sparse signal $a e_{j_i}$. Thus, in the second step, we can retrieve the value of the nonzero entry, equal to $\wh{u}^{(0)}_{i}$, as well as the index $j_i$ from the phase of the ratio $\wh{u}^{(1)}_{i} /\wh{u}^{(0)}_{i}$ (this technique was introduced in~\cite{HIKP2,LWC} and was referred to as the ``OFDM trick''). The total time is dominated by the cost of the two DFTs of the columns, which is $O(\sqrt{n} \log n)$. Since the algorithm queries only a constant number of columns, its sample complexity is $O(\sqrt{n})$.
 
 In general, the distribution of the nonzero entries over the rows can be non-uniform. Thus, our actual algorithm alternates the above recovery process between the columns and rows (see Figure~\ref{fig:innerloop} for an illustration).  Since the OFDM trick works only on $1$-sparse columns/rows, we check the $1$-sparsity of each column/row by sampling a constant number of additional entries.  We  then show that, as long as the sparsity constant $a$ is small enough, this process recovers all entries in a logarithmic number steps with constant probability. The proof uses the fact that the probability of the existence of an ``obstructing configuration'' of nonzero entries which makes the process deadlocked (e.g., see Figure~\ref{fig:deadlock}) is upper bounded by a small constant.


\begin{figure*}
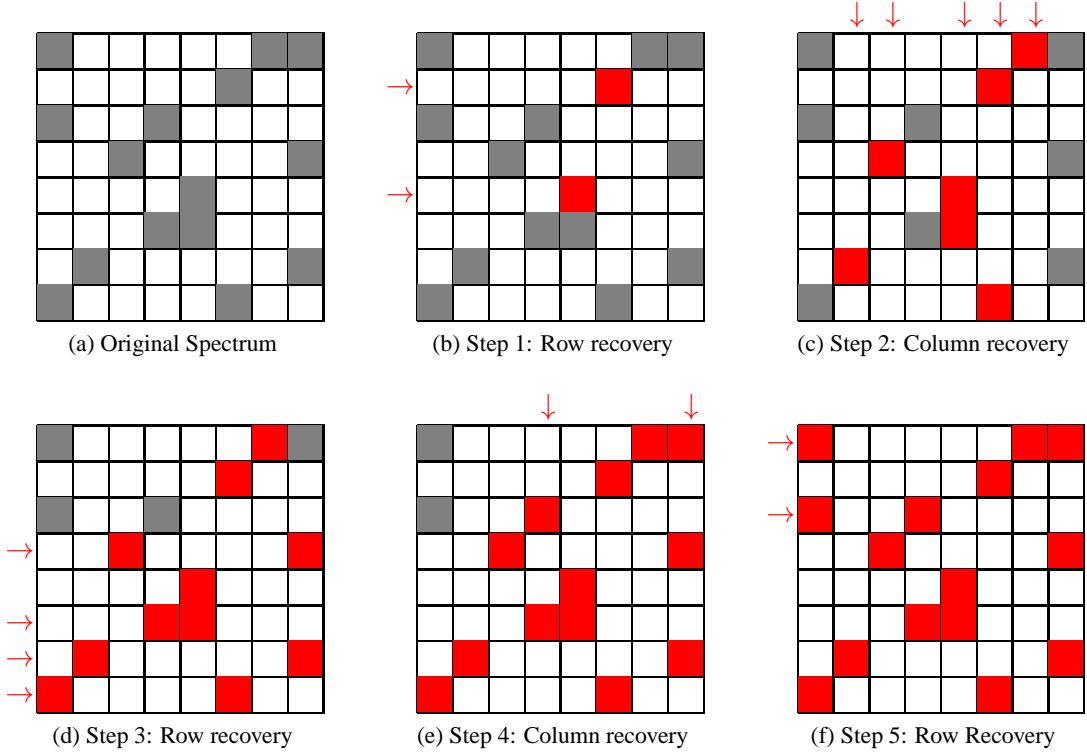

  \centering
  \subfloat[][Original Spectrum]{
    \begin{tabular}{p{0.2cm}|p{0.04cm}|p{0.04cm}|p{0.04cm}|p{0.04cm}|p{0.04cm}|p{0.04cm}|p{0.04cm}|p{0.04cm}|l}
        \multicolumn{10}{l}{} \\
        \cline{2-9}
         & \cellcolor{gray}  &  &  &  &  &  & \cellcolor{gray} & \cellcolor{gray} &  \\ \cline{2-9}
         & &  &  &  &  & \cellcolor{gray} &  & &  \\ \cline{2-9} 
         & \cellcolor{gray} &  & &\cellcolor{gray}   &  &  &  & &  \\ \cline{2-9}
         &  &  & \cellcolor{gray}  &  &  &  &  & \cellcolor{gray} &  \\ \cline{2-9} 
         & &  &  &  &  \cellcolor{gray} &  & &  &  \\ \cline{2-9} 
         & &  &  &\cellcolor{gray}  &  \cellcolor{gray} & &  & & \\ \cline{2-9} 
         & & \cellcolor{gray} &  &  &  &  &  &\cellcolor{gray} & \\ \cline{2-9} 
         & \cellcolor{gray}  &  &  &  &  & \cellcolor{gray} &  & & \\ \cline{2-9} 
    \end{tabular}
  }
  \subfloat[][Step 1: Row recovery]{
    \begin{tabular}{p{0.2cm}|p{0.04cm}|p{0.04cm}|p{0.04cm}|p{0.04cm}|p{0.04cm}|p{0.04cm}|p{0.04cm}|p{0.04cm}|l}
        \multicolumn{10}{l}{} \\
        \cline{2-9}
        & \cellcolor{gray} &  &  &  &  &  & \cellcolor{gray} & \cellcolor{gray} &  \\ \cline{2-9}
        \textred{$\rightarrow$} & &  &  &  &  & \cellcolor{red} &  & &  \\ \cline{2-9}
        & \cellcolor{gray} &  & &\cellcolor{gray}   &  &  &  &  &   \\ \cline{2-9}
        & &  & \cellcolor{gray}  &  &  &  &  & \cellcolor{gray} &    \\ \cline{2-9}
        \textred{$\rightarrow$}& &  &  &  &   \cellcolor{red} & & & & \\ \cline{2-9} 
        & &  &  &\cellcolor{gray}  &   \cellcolor{gray} & &  & &     \\ \cline{2-9}
        & & \cellcolor{gray} &  &  &  &  &  &\cellcolor{gray}  &    \\ \cline{2-9}
        & \cellcolor{gray}  &  &  &  &  & \cellcolor{gray} &  & & \\ \cline{2-9}
    \end{tabular}
  }
  \subfloat[][Step 2: Column recovery]{
    \begin{tabular}{p{0.2cm}|p{0.04cm}|p{0.04cm}|p{0.04cm}|p{0.04cm}|p{0.04cm}|p{0.04cm}|p{0.04cm}|p{0.04cm}|l}
        \multicolumn{2}{l}{} & \multicolumn{1}{p{0.04cm}}{\textred{$\downarrow$}} &\multicolumn{1}{p{0.04cm}}{\textred{$\downarrow$}} & \multicolumn{1}{l}{} &\multicolumn{1}{p{0.04cm}}{\textred{$\downarrow$}}  &\multicolumn{1}{p{0.04cm}}{\textred{$\downarrow$}} & \multicolumn{1}{p{0.04cm}}{\textred{$\downarrow$}} & \multicolumn{2}{l}{}   \\
        \cline{2-9}
        & \cellcolor{gray} &  &  &  &  &  & \cellcolor{red} & \cellcolor{gray} &\\ \cline{2-9}
        &  &  &  &  &  & \cellcolor{red} &  & &  \\ \cline{2-9}
        & \cellcolor{gray} &  & &\cellcolor{gray}   &  &  &  & &   \\ \cline{2-9}
        &  &  & \cellcolor{red}  &  &  &  &  & \cellcolor{gray} &   \\ \cline{2-9}
        &  &  &  &  &   \cellcolor{red} & & & &   \\ \cline{2-9}
        &  &  &  &\cellcolor{gray}  &   \cellcolor{red} & &  &  &    \\ \cline{2-9}
        &  & \cellcolor{red} &  &  &  &  &  &\cellcolor{gray} &    \\ \cline{2-9}
        & \cellcolor{gray}  &  &  &  &  & \cellcolor{red} &  & &  \\ \cline{2-9}
    \end{tabular}
  }
  \qquad
  \subfloat[][Step 3: Row recovery]{
    \begin{tabular}{p{0.2cm}|p{0.04cm}|p{0.04cm}|p{0.04cm}|p{0.04cm}|p{0.04cm}|p{0.04cm}|p{0.04cm}|p{0.04cm}|l}
        \multicolumn{10}{l}{} \\
        \cline{2-9}
        & \cellcolor{gray} &  &  &  &  &  & \cellcolor{red} & \cellcolor{gray} &\\ \cline{2-9}
        &  &  &  &  &  & \cellcolor{red} &  & &  \\ \cline{2-9}
        & \cellcolor{gray} &  & &\cellcolor{gray}   &  &  &  & &   \\ \cline{2-9}
       \textred{$\rightarrow$}  &  &  & \cellcolor{red}  &  &  &  &  & \cellcolor{red}  &   \\ \cline{2-9}
        &  &  &  &  &   \cellcolor{red} & & &  &   \\ \cline{2-9}
       \textred{$\rightarrow$}  &  &  &  &\cellcolor{red}  &   \cellcolor{red} & &  &  &    \\ \cline{2-9}
       \textred{$\rightarrow$}  &  & \cellcolor{red} &  &  &  &  &  &\cellcolor{red}  &    \\ \cline{2-9}
       \textred{$\rightarrow$}  & \cellcolor{red}  &  &  &  &  & \cellcolor{red} &  & &   \\ \cline{2-9}
    \end{tabular}
  }
  \subfloat[][Step 4: Column recovery]{
    \begin{tabular}{p{0.2cm}|p{0.04cm}|p{0.04cm}|p{0.04cm}|p{0.04cm}|p{0.04cm}|p{0.04cm}|p{0.04cm}|p{0.04cm}|l}
        \multicolumn{4}{l}{} &\multicolumn{1}{p{0.04cm}}{\textred{$\downarrow$}}  &\multicolumn{3}{l}{} & \multicolumn{1}{p{0.04cm}}{\textred{$\downarrow$}} & \multicolumn{1}{l}{}   \\
        \cline{2-9}
        & \cellcolor{gray} &  &  &  &  &  & \cellcolor{red} & \cellcolor{red} & \\ \cline{2-9}
        & &  &  &  &  & \cellcolor{red} &  &  & \\ \cline{2-9}
        & \cellcolor{gray} &  & &\cellcolor{red}   &  &  &  &  &  \\ \cline{2-9}
        & &  & \cellcolor{red}  &  &  &  &  & \cellcolor{red}  &  \\ \cline{2-9}
        & &  &  &  &   \cellcolor{red} & & & &   \\ \cline{2-9}
        & &  &  &\cellcolor{red}  &   \cellcolor{red} & &  &  &    \\ \cline{2-9}
        &  & \cellcolor{red} &  &  &  &  &  &\cellcolor{red}   &  \\ \cline{2-9}
        & \cellcolor{red}  &  &  &  &  & \cellcolor{red} &  &  & \\ \cline{2-9}
    \end{tabular}
  }
  \subfloat[][Step 5: Row Recovery]{
    \begin{tabular}{p{0.2cm}|p{0.04cm}|p{0.04cm}|p{0.04cm}|p{0.04cm}|p{0.04cm}|p{0.04cm}|p{0.04cm}|p{0.04cm}|l}
        \multicolumn{10}{l}{} \\
        \cline{2-9}
       \textred{$\rightarrow$}  & \cellcolor{red} &  &  &  &  &  & \cellcolor{red} & \cellcolor{red} &  \\ \cline{2-9}
         &   &  &  &  &  & \cellcolor{red} &  & &    \\ \cline{2-9}
       \textred{$\rightarrow$}  &  \cellcolor{red} &  & &\cellcolor{red}   &  &  &  & &     \\ \cline{2-9}
        &   &  & \cellcolor{red}  &  &  &  &  &   \cellcolor{red} &  \\ \cline{2-9}
        &   &  &  &  &   \cellcolor{red} & & & &     \\ \cline{2-9}
        &   &  &  &\cellcolor{red}  &   \cellcolor{red} & &  &  &      \\ \cline{2-9}
        &   & \cellcolor{red} &  &  &  &  &  &\cellcolor{red}&       \\ \cline{2-9}
        &  \cellcolor{red}  &  &  &  &  & \cellcolor{red} &  & &    \\ \cline{2-9}
    \end{tabular}
  }
  \caption{An illustration of the ``peeling'' recovery process on an $8\times8$ signal with 15 nonzero frequencies. In each step, the algorithm recovers all $1$-sparse columns and rows (the recovered entries are depicted in red). The process converges after a few steps.}
  \label{fig:innerloop}
\end{figure*}

\begin{figure*}
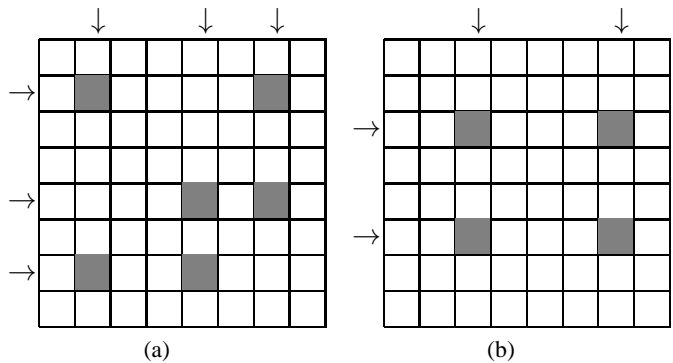

  \centering
    \subfloat[][]{
    \begin{tabular}{p{0.2cm}|p{0.04cm}|p{0.04cm}|p{0.04cm}|p{0.04cm}|p{0.04cm}|p{0.04cm}|p{0.04cm}|p{0.04cm}|l}
       \multicolumn{2}{l}{} &\multicolumn{1}{p{0.04cm}}{$\downarrow$}  &\multicolumn{2}{l}{} & \multicolumn{1}{p{0.04cm}}{$\downarrow$} & \multicolumn{1}{l}{} &\multicolumn{1}{p{0.04cm}}{$\downarrow$} & \multicolumn{2}{l}{} \\
        \cline{2-9}
         &  &  &  &  &  &  &   &   \\ \cline{2-9}
         $\rightarrow$ &  & \cellcolor{gray}   &  &  &  & & \cellcolor{gray} &  \\ \cline{2-9}
         &  &  &  &  &  &  & &  \\ \cline{2-9}
         &  & & &  &  &  &   &  \\ \cline{2-9}
         $\rightarrow$ &  &  &  &  & \cellcolor{gray} &  & \cellcolor{gray}  &   \\ \cline{2-9}
         &  &  &  & & & &  &\\ \cline{2-9}
         $\rightarrow$ &  & \cellcolor{gray}    &  &  & \cellcolor{gray} &  & &  \\ \cline{2-9}
         &  &  &  &  &  &  &   &   \\ \cline{2-9}
  \end{tabular}
  }
   \subfloat[][]{
    \begin{tabular}{p{0.2cm}|p{0.04cm}|p{0.04cm}|p{0.04cm}|p{0.04cm}|p{0.04cm}|p{0.04cm}|p{0.04cm}|p{0.04cm}|l}
       \multicolumn{3}{l}{} &\multicolumn{1}{p{0.04cm}}{$\downarrow$}  &\multicolumn{3}{l}{} &\multicolumn{1}{p{0.04cm}}{$\downarrow$} & \multicolumn{2}{l}{} \\
        \cline{2-9}
         &  &  &  &  &  &  &   &   \\ \cline{2-9}
         &  &  &  &  &  &  &   &   \\ \cline{2-9}
         $\rightarrow$ & &   & \cellcolor{gray}   &  &  & & \cellcolor{gray}    &    \\ \cline{2-9}
         &  &  &  &  &  &  & &  \\ \cline{2-9}
         &  & & &  &  &  &   &  \\ \cline{2-9}
         $\rightarrow$ &  &   & \cellcolor{gray}   &  &  & & \cellcolor{gray}   &     \\ \cline{2-9}
         &  &   &  &  & &  & &  \\ \cline{2-9}
         &  &  &  &  &  &  &   &   \\ \cline{2-9}

  \end{tabular}
  }
  \qquad

  \caption{Examples of obstructing sequences of nonzero entries. None of the remaining rows or columns has a sparsity of 1.}
  \label{fig:deadlock}
\end{figure*}

The algorithm is extended to the case of $k=o(\sqrt{n})$ via a reduction. 
Specifically, we subsample the signal $x$ by the reduction ratio
$R=\alpha\sqrt{n}/k$ for some small enough constant $\alpha$ in each dimension.
The subsampled signal $x'$ has dimension $\sqrt{m}\times\sqrt{m}$,
where $\sqrt{m} = \frac{k}{\alpha}$.  Since subsampling in time domain
corresponds to ``spectrum folding'', i.e., adding together all
frequencies with indices that are equal modulo $\sqrt{m}$, the nonzero
entries of $\wh{x}$ are mapped into the entries of $\wh{x'}$. It can
be seen that, with constant probability, the mapping is one-to-one. If
this is the case, we can use the earlier algorithm for sparse DFT to
compute the nonzero frequencies in $O(\sqrt{m} \log m) = O(\sqrt{k}
\log k)$ time, using $O(k)$ samples. We then use the OFDM trick to
identify the positions of those frequencies in $\wh{x}$.

Our second algorithm for the exactly sparse case works for all values
of $k$. The main idea behind it is to decode rows/columns with higher
sparsity than $1$.  First, we give a \emph{deterministic},
\emph{worst-case} algorithm for 1-dimensional sparse Fourier
transforms that takes $O(k^2 + k (\log \log n)^{O(1)})$ time.  This
algorithm uses the relationship between sparse recovery and
syndrome decoding of Reed-Solomon codes (due to
\cite{reed-solomon}). Although a simple application of the decoder
yields $O(n^2)$ decoding time, we show that by using appropriate
numerical subroutines one can in fact recover a $k$-sparse vector from
$O(k)$ samples in time $O(k^2 + k (\log\log n)^{O(1)})$\footnote{We
  note that, for $k = o(\log{n})$, this is the fastest known {\em
    worst-case} algorithm for the exactly sparse DFT.}.  In
particular, we use Berlekamp-Massey's algorithm for constructing the
error-locator polynomial and Pan's algorithm for finding its roots.
For our fast average-case, $2$-dimensional sparse Fourier transform
algorithm, we fold the spectrum into $B=\frac{k}{C \log k}$ bins for
some large constant $C$. Since the positions of the $k$ nonzero
frequencies are random, it follows that each bin receives
$t=\Theta(\log k)$ frequencies with high probability. We then take
$\Theta(t)$ samples of the time domain signal corresponding to each
bin, and recover the frequencies corresponding to those bins in $O(t^2
+ t (\log\log n)^{O(1)})$ time per bin, for a total time of $O(k \log
k + k (\log\log n)^{O(1)})$.

The above approach works as long as the number of nonzero coefficients
per column/row are highly concentrated. However, this is not the case
for $k \ll \sqrt{n} \log n$. We overcome this difficulty by replacing
a row by a sequence of rows. A technical difficulty is that the
process might lead to collisions of coefficients. We resolve this
issue by using a two level procedure, where the first level returns
the syndromes of colliding coefficients as opposed to the coefficients
themselves; the syndromes are then decoded at the second level.\xxx{We
  need this paragraph, since the algorithms are in the
  appendix. Hopefully the description makes sense.}

Our third algorithm works for \emph{approximately} sparse data, at
sparsity $\Theta(\sqrt{n})$.  Its general outline mimics that of the
first algorithm. Specifically, it alternates between decoding columns
and rows, assuming that they are $1$-sparse. The decoding subroutine
itself is similar to that of \cite{HIKP2} and uses $O(\log n)$
samples. The subroutine first checks whether the decoded entry is
large; if not, the spectrum is unlikely to contain any large entry,
and the subroutine terminates. The algorithm then subtracts the
decoded entry from the column and checks whether the resulting signal
contains no large entries in the spectrum (which would be the case if
the original spectrum was approximately $1$-sparse and the decoding
was successful). The check is done by sampling $O(\log n)$ coordinates
and checking whether their sum of squares is small. To prove that this
check works with high probability, we use the fact that a collection
of random rows of the Fourier matrix is likely to satisfy the
Restricted Isometry Property (RIP) of~\cite{CTao}.

A technical difficulty in the analysis of the algorithm is that the
noise accumulates in successive iterations.  This means that a
$1/\log^{O(1)} n$ fraction of the steps of the algorithm will fail.
However, we show that the dependencies are ``local'', which means that
our analysis still applies to a vast majority of the recovered
entries. We continue the iterative decoding for $\log \log n$ steps,
which ensures that all but a $1/\log^{O(1)} n$ fraction of the large
frequencies are correctly recovered. To recover the remaining
frequencies, we resort to algorithms with worst-case guarantees.

\paragraph{Extensions} Our algorithms have natural extensions to dimensions higher than $2$. We do not include them in this paper as the description and analysis are rather cumbersome. 

While no optimal result is known for the $1$-dimensional case, one can achieve
optimal sample complexity and efficient robust
recovery in the $\log n$-dimensional (Hadamard) case (\cite{Lev}, see also Appendix C.2 of~\cite{Gold}).
Our result demonstrates that even two dimensions give enough
flexibility for optimal sample complexity in the average case.  Due to
the equivalence between the two-dimensional case and the one-dimensional
case where $n$ is a product of different prime powers~\cite{GMS,Iw-arxiv}, our
algorithm also gives optimal sample complexity bounds for e.g., $n = 6^t$
in the average case.
\\
\\

\subsection{Related work}

As described in the introduction, currently the most efficient
algorithms for computing the sparse DFT are due to~\cite{HIKP2}.  For
signals that are exactly $k$-sparse, the first algorithm runs in $O(k
\log n)$ time.  For approximately sparse signals, the second algorithm
runs in $O(k \log n \log(n/k))$ time.  Formally, the latter algorithm
works for any signal $x$, and computes an approximation vector
$\wh{x}'$ that satisfies the $\ell_2/\ell_2$ approximation guarantee,
i.e., $\| \wh{x}-\wh{x}'\|_2 \le C \min_{k \text{-sparse } y }
\|\wh{x}-y\|_2$, where $C$ is some approximation factor and the
minimization is over $k$-sparse signals.  Note that this guarantee
generalizes that of Equation~(\ref{eq:guarantee}).

We also mention another efficient algorithm, due to~\cite{LWC},
designed for the exactly $k$-sparse model. The average case analysis
presented in that paper shows that the algorithm has $O(k)$ expected
sample complexity and runs in $O(k \log k)$ time. However, the
algorithm assumes that the input signal $x$ is specified as a {\em
  function} over an interval $[0,1]$ that can be sampled at arbitrary
positions, as opposed to a given discrete sequence of $n$ samples as
in our case. Thus, although very efficient, that algorithm does not
solve the Discrete Fourier Transform problem.

\newcommand{\hs}{h_{\sigma,b}}
\newcommand{\hsr}{h_{\sigma_r,b_r}}
\newcommand{\ps}{\pi_{\sigma,b}}
\newcommand{\os}{o_{\sigma,b}}
\newcommand{\comments}[1]{}

\section{Preliminaries}\label{sec:defs}
This section introduces the notation, assumptions and definitions
used in the rest of this paper.


\paragraph{Notation} Throughout the paper we assume that $\sqrt{n}$ is a power of $2$. We use $[m]$ to denote the set $\{ 0, \dotsc,
m-1\}$, and $[m]\times[m] = [m]^2$ to denote the $m\times m$ grid $\{(i,j) : i \in [m], j \in [m]\}$.
We define $\omega = e^{-2\pi \mathbf{i} / \sqrt{n}}$ to be a primitive $\sqrt{n}$-th
root of unity and $\omega' = e^{-2\pi \mathbf{i} / n}$ to be a primitive $n$-th
root of unity.  For any complex number $a$, we use $\phi(a) \in [0,
2\pi)$ to denote the {\em phase} of $a$.  \xxx{For a complex number $a$ and
a real positive number $b$, the expression $a \pm b$ denotes a complex
number $a'$ such that $\abs{a-a'} \le b$.}  For a 2D matrix $x \in \C^{\sqrt{n}\times \sqrt{n}}$,
its support is denoted by $\supp(x) \subseteq [\sqrt{n}]\times [\sqrt{n}]$.  We use
$\norm{0}{x}$ to denote $\abs{\supp(x)}$, the number of nonzero
coordinates of $x$.  Its 2D Fourier spectrum is denoted by $\wh{x}$, with
\[
\wh{x}_{i,j} = \frac{1}{\sqrt{n}}\sum_{l\in [\sqrt{n}]}{\sum_{m\in [\sqrt{n}]}{ \omega^{il+jm} x_{l,m}}}.
\]
Similarly, if $y$ is a frequency-domain signal, its inverse Fourier 
transform is denoted by $\check{y}$.

\paragraph{Definitions} 
The paper uses the comb filter used in~\cite{Iw,HIKP} (cf.~\cite{Man}). The filter can be generalized to $2$ dimensions as follows: 

Given $(\tau_r,\tau_c) \in [\sqrt{n}]\times[\sqrt{n}]$, and $B_r, B_c$
that divide $\sqrt{n}$, then for all $(i,j) \in[B_r]\times[B_c]$ set
\[
y_{i,j} = x_{i(\sqrt{n}/B_r) + \tau_r, j(\sqrt{n}/B_c) + \tau_c}.
\]
Then, compute the 2D DFT $\hat{y}$ of $y$. Observe that  $\hat{y}$ is a folded version of $\hat{x}$:
\[\hat{y}_{i,j} = \displaystyle \sum_{l \in [\sqrt{n}/B_r]}{\sum_{m \in [\sqrt{n}/B_c]}{ \hat{x}_{l B_r + i, m B_c + j} \omega^{-\tau_r(i + l B_r)-\tau_c(j + mB_c)}}}.
\]

\comments{
Thus,
$
\E_{\tau_r,\tau_c}[\sabs{\hat{y}_{i,j}}^2] =\displaystyle \sum_{l  \equiv j \bmod B_r}{\sum_{m \equiv j \bmod B_c}{\sabs{\hat{x}_{l,m}}^2}}.
$}

\paragraph{Distributions}
In the exactly sparse case,  we assume a Bernoulli model for the support of $\wh{x}$. This 
means that for all $(i,j) \in [\sqrt{n}]\times[\sqrt{n}]$, 
\[
\mbox{Pr}\{(i,j) \in \supp{(\wh{x})}\} = k/n 
\]
and thus $\E[\abs{\supp{(\wh{x})}}] = k$.
We assume an unknown predefined matrix $a_{i,j}$ of values in $\C$; if $\wh{x}_{i,j}$ is selected to be nonzero, its value is set to $a_{i,j}$.

\xxx{
For our second algorithm, we assume that $\wh{x}_{i,j} \in
  \{-L, \ldots, L\}$ for some precision parameter $L$. To simplify the
  bounds, we assume that $L =n^{O(1)}$; otherwise the $\log n$ term in
  the running time bound is replaced by $\log L$.}
\xxx{\textred{This doesn't seem needed; Pan's algorithm is needed to recover the frequency
positions, there is no assumption for the values.}}

In the approximately sparse case, we assume that the signal $\wh{x}$ is equal to $ \wh{x^*} + \wh{w} \in \C^{\sqrt{n} \times
  \sqrt{n}}$, where $\wh{x^*}_{i,j}$ is the ``signal'' and $\wh{w}$ is
the ``noise''.  In particular, $\wh{x^*}$ is drawn from the Bernoulli model,
where $\wh{x^*}_{i,j}$ is drawn from $\{0, a_{i,j}\}$ at
random independently for each $(i,j)$ for some values $a_{i,j}$ and
with $\E[|\supp(\wh{x^*})|] = k$.  We also require that
$\abs{a_{i,j}} \geq L$ for some parameter $L$.  $\wh{w}$ is a complex Gaussian
vector with variance $\sigma^2$ in both the real and imaginary axes
independently on each coordinate; we notate this as $\wh{w} \sim N_{\C}(0, \sigma^2I_n)$.  We will need that $L = C \sigma\sqrt{n/k}$
for a sufficiently large constant $C$, so that $\E[\norm{2}{\wh{x^*}}^2] \ge C
\E[\norm{2}{\wh{w}}^2]$.

\xxx{Something about precision/model}

\section{Basic Algorithm for the Exactly Sparse Case}\label{sec:inf}
The algorithm for the noiseless case depends on the sparsity $k$ where
$k =\E[\abs{\supp{(\wh{x})}}]$ for a Bernoulli distribution of the
support.

\subsection{Basic Exact Algorithm: $k = \Theta(\sqrt{n})$}
\label{subsec:basic_exact}

In this section, we focus on the regime $k = \Theta(\sqrt{n})$. Specifically, we will assume that $k = a\sqrt{n}$ for a (sufficiently small) constant $a>0$.

The algorithm \textsc{BasicExact2DSFFT} is described as Algorithm~\ref{a:basicexact}. 
The key idea is to fold the spectrum into bins using the comb filter 
defined in~\xref{sec:defs} and estimate frequencies which are isolated in a bin.
The algorithm takes the FFT of a row and as a result frequencies 
in the same columns will get folded into the same row bin. It
also takes the FFT of a column and consequently frequencies
in the same rows wil get folded into the same column bin. The algorithm
then uses the OFDM trick introduced in~\cite{HIKP2} to recover the columns and rows whose sparsity is 1.
It iterates between the column bins and row bins, subtracting the recovered frequencies
and estimating the remaining columns and rows whose sparsity is 1. An illustration of the 
algorithm running on an $8\times8$ signal with 15 nonzero frequencies is shown in Fig.~\ref{fig:innerloop} in Section 1. 
The algorithm also takes a constant number 
of extra FFTs of columns and rows to check for collisions within a bin and avoid errors resulting from 
estimating bins where the sparsity is greater than 1. The algorithm uses three functions:
\begin{itemize}
\item \textsc{FoldToBins}. This procedure folds the spectrum into $B_r \times B_c$ bins using the comb filter described~\xref{sec:defs}.
\item \textsc{BasicEstFreq}. Given the FFT of rows or columns, it estimates the frequency in the large bins. If there is no collision, i.e. if there is a single nonzero frequency in the bin, it adds this frequency to the result $\wh{w}$ and subtracts its contribution to the row and column bins. 
\item \textsc{BasicExact2DSFFT}. This performs the FFT of the rows and columns and then iterates \textsc{BasicEstFreq} between the rows and columns until is recovers $\wh{x}$.
\end{itemize}

\begin{algorithm}[!ht]
 \caption{Basic Exact 2D sparse FFT algorithm for $k = \Theta(\sqrt{n})$}\label{a:basicexact}
  \begin{algorithmic}
    \Procedure{FoldToBins}{$x$, $B_r$, $B_c$, $\tau_r$, $\tau_c$}
    \State $y_{i,j} = x_{i(\sqrt{n}/B_r) + \tau_r, j(\sqrt{n}/B_c) + \tau_c} $ for $(i,j) \in [B_r]\times[B_c]$,
    \State \Return $\wh{y}$, the DFT of $y$
    \EndProcedure
    \Procedure{BasicEstFreq}{$\wh{u}^{(T)}$, $\wh{v}^{(T)}$,$T$, IsCol}
    \State $\wh{w} \gets 0$.
    \State Compute $J=\{j: \sum_{\tau\in{T}}{|\wh{u}^{(\tau)}_j|} > 0\}$. 
    \For{$j \in J$}
    \State $b \gets \wh{u}^{(1)}_{j} / \wh{u}^{(0)}_{j}$.
    \State $i  \gets \text{round}(\phi(b)\frac{\sqrt{n}}{2\pi}) \bmod \sqrt{n}$. \Comment{$\phi(b)$ is the phase of $b$.}
    \State $s \gets \wh{u}^{(0)}_{j}$.
    \State  \Comment{Test whether the row or column is 1-sparse}
    \If{$\left(\sum_{\tau\in{T}}{|\wh{u}^{(\tau)}_j - s\omega^{-\tau i}|} == 0 \right)$}
        \If{IsCol}  \Comment {whether decoding column or row}
            \State $\wh{w}_{i,j} \gets s$.
        \Else
            \State $\wh{w}_{j,i} \gets s$.
        \EndIf
        \For{$\tau\in{T}$}
        \State $\wh{u}^{(\tau)}_j \gets 0$
        \State $\wh{v}^{(\tau)}_i \gets \wh{v}^{(\tau)}_i -  s\omega^{-\tau i}$
        \EndFor
   
    \EndIf
    \EndFor
    \State \Return $\wh{w}$, $\wh{u}^{(T)}$, $\wh{v}^{(T)}$
    \EndProcedure
    \Procedure{BasicExact2DSFFT}{$x$, $k$}
    \State $T \gets [2c]$  \Comment{We set $c\ge 6$}
    \For{$\tau \in T$}
    \State $\wh{u}^{(\tau)} \gets \textsc{FoldToBins}(x, \sqrt{n}, 1, 0, \tau)$.
    \State $\wh{v}^{(\tau)} \gets \textsc{FoldToBins}(x, 1, \sqrt{n}, \tau, 0)$.
    \EndFor
    \State $\wh{z} \gets 0$
    \For{$t \in [C\log n]$}\Comment{$\wh{u}^{(T)}:=\{\wh{u}^{(\tau)}\,:\,\tau\in{T}\}$}
    \State $\{\wh{w}, \wh{u}^{(T)}, \wh{v}^{(T)}\} \gets \textsc{BasicEstFreq}(\wh{u}^{(T)},  \wh{v}^{(T)}, T, $ true). 
    \State $\wh{z} \gets \wh{z}  + \wh{w}$.
    \State $\{\wh{w}, \wh{v}^{(T)}, \wh{u}^{(T)}\}  \gets \textsc{BasicEstFreq}(\wh{v}^{(T)},  \wh{u}^{(T)}, T, $ false).
    \State $\wh{z} \gets \wh{z}  + \wh{w}$.
    \EndFor
    \State \Return $\wh{z}$
    \EndProcedure
  \end{algorithmic}
\end{algorithm}

\paragraph{Analysis of \textsc{BasicExact2DSFFT}}
\begin{lemma}
\label{l:bc}
For any constant $\alpha > 0$, if $a > 0$ is a sufficiently small constant,
then assuming that all 1-sparsity tests in the procedure
\textsc{BasicEstFreq} are correct, the algorithm reports the correct
output with probability at least $1-O(\alpha)$.
\end{lemma}

\begin{proof}
  The algorithm fails if there is a pair of nonzero entries in a
  column or row of $\wh{x}$ that ``survives'' $t_{max}=C\log n$
  iterations.  For this to happen there must be an ``obstructing''
  sequence of nonzero entries $p_1, q_1, p_2, q_2 \ldots p_t$, $3 \le
  t \le t_{max}$, such that for each $i \ge1$, $p_i$ and $q_i$ are in
  the same column (``vertical collision''), while $q_i$ and $p_{i+1}$
  are in the same row (``horizontal collision''). Moreover, it must be
  the case that either the sequence ``loops around'', i.e., $p_1=p_t$,
  or $t>t_{max}$. We need to prove that the probability of either case
  is less than $\alpha$. We focus on the first case; the second one is
  similar.

  Assume that there is a sequence $p_1, q_1, \ldots p_{t-1}, q_{t-1}$
  such that the elements in this sequence are all distinct, while
  $p_1=p_t$.  If such a sequence exists, we say that the event $E_t$
  holds.  The number of sequences satisfying $E_t$ is at most
  $\sqrt{n}^{2(t-1)}$, while the probability that the entries
  corresponding to the points in a specific sequence are nonzero is
  at most $(k/n)^{2(t-1)} = (a/\sqrt{n})^{2(t-1)}$. Thus the
  probability of $E_t$ is at most
  \[
  \sqrt{n}^{2(t-1)} \cdot (a/\sqrt{n})^{2(t-1)} = a^{2(t-1)}.
  \]
  Therefore, the probability that one of the events $E_1, \ldots,
  E_{t_{max}}$ holds is at most $\sum_{t=3}^\infty a^{2(t-1)} =
  a^4/(1-a^2)$, which is smaller than $\alpha$ for $a$ small enough.
\end{proof}

\begin{lemma}
\label{l:1sparse}
The probability that any 1-sparsity test invoked by the algorithm is incorrect is at most $O(1/n^{(c-5)/2})$.
\end{lemma}

To prove Lemma~\ref{l:1sparse}, we first need the following lemma.

\begin{lemma}
\label{l:test}
Let $y \in \C^m$ be drawn from a permutation invariant distribution
with $r \geq 2$ nonzero values.  Let $T = [2c]$.  Then the
probability that there exists a $y'$ such that $\norm{0}{y'} \leq 1$
and $(\wh{y}-\wh{y}')_T = 0$ is at most
$c\left(\frac{c}{m-r}\right)^{c-2}$.
\end{lemma}
\begin{proof}

  Let $A = F_T$ be the first $2c$ rows of the inverse Fourier matrix.  Because
  any $2c \times 2c$ submatrix of $A$ is Vandermonde and hence
  non-singular, the system of linear equations
  \[
  Az = b
  \]
  has at most one $c$-sparse solution in $z$, for any $b$.

  If $r \leq c-1$, then $\norm{0}{y-y'} \leq c$ so $A(y-y') = 0$
  implies $y - y' = 0$.  But $r \geq 2$ so $\norm{0}{y-y'} > 0$.  This
  is a contradiction, so if $r < c$ then the probability that
  $(\wh{y}-\wh{y}')_T = 0$ is zero.  Henceforth, we assume $r \geq c$.

  When drawing $y$, first place $r-(c-1)$ coordinates into $u$ then
  place the other $c-1$ values into $v$, so that $y = u + v$.
  Condition on $u$, so $v$ is a permutation distribution over
  $m-r+c-1$ coordinates.  We know there exists at most one $c$-sparse
  vector $w$ with $Aw = -Au$.  Then

  \begin{align*}
    &  \Pr_y[\exists y' : A(y-y') = 0 \andop \norm{0}{y'} \leq 1] \\
    = & \Pr_v[\exists y' : A(v-y') = -Au \andop \norm{0}{y'} \leq 1]  \\
    \leq & \Pr_v[\exists y' : v-y' = w \andop \norm{0}{y'} \leq 1] = \Pr_v[\norm{0}{v-w} \leq 1]  \\
    \leq &  \Pr_v[\abs{\supp(v) \symdiff \supp(w)} \leq 1]\\
    < &\frac{m-r+c-1}{\binom{m-r+c-1}{c-1}} < c\left(\frac{c}{m-r}\right)^{c-2}
  \end{align*}
  where the penultimate inequality follows from considering the cases
  $\norm{0}{w} \in \{c-2, c-1, c\}$ separately.
\end{proof}

We now proceed with the proof of  Lemma~\ref{l:1sparse} .
\begin{proof}
W.L.O.G. consider the row case. Let $y$ be the $j$th row of $\wh{x}$.
Note that $\wh{u}^{(\tau)}_{j} =\wh{y}_{\tau}$.  Observe that with
probability at least $1-1/n^c$ we have $\|y\|_0 \le r$ for $r=c \log
n$.  Moreover, the distribution of $y$ is permutation-invariant, and
the test in \textsc{BasicEstFreq} corresponds to checking whether
$(\wh{y}-\wh{y}')_T=0$ for some $1$-sparse $y'=a e_i$. Hence,
Lemma~\ref{l:test} with $m = \sqrt{n}$ implies the probability that
any specific test fails is less than $c(2c/\sqrt{n})^{c-2}$. 
Taking a
union bound over the $\sqrt{n} \log n$ total tests gives a failure
probability of $4c^3\log n (2c/\sqrt{n})^{c-4}  < O(1/n^{(c-5)/2})$.
\end{proof}

\begin{theorem}\label{th:basic}
  For any constant $\alpha$, the algorithm
  \textsc{BasicExact2DSFFT} uses $O(\sqrt{n})$ samples, runs in
  time $O(\sqrt{n} \log n)$ and returns the correct vector $\wh{x}$
  with probablility at least $1-O(\alpha)$ as long as $a$ is a small
  enough constant.
\end{theorem}

\begin{proof}
  From Lemma~\ref{l:bc} and Lemma~\ref{l:1sparse}, the algorithm
  returns the correct vector $\wh{x}$ with probability at least $1 -
  O(\alpha) - O(n^{-(c-5)/2}) = 1-O(\alpha)$ for $c > 5$.

The algorithm uses only $O(T)=O(1)$ rows and columns of $x$, which yields $O(\sqrt{n})$ samples. The running time is bounded by the time needed to perform $O(1)$ FFTs of rows and columns (in \textsc{FoldToBins}) procedure, and $O(\log n)$ invocations of  \textsc{BasicEstFreq}. Both components take time $O(\sqrt{n} \log n)$.

\end{proof}

    




\subsection{Reduction to Basic Exact Algorithm: $k = o(\sqrt{n})$}

Algorithm \textsc{ReduceExact2DSFFT}, which is for the case where $k=o(\sqrt{n})$, is described in Algorithm~\ref{a:exact}).
The key idea is to reduce the problem from the case where $k=o(\sqrt{n})$ to the case where $k=\Theta(\sqrt{n})$.
To do that, we subsample the input time domain signal $x$ by the reduction ratio $R=a\sqrt{n}/k$ for some small enough $a$. 
The subsampled signal $x'$ has dimension $\sqrt{m}\times\sqrt{m}$, where $\sqrt{m} = \frac{k}{a}$. This implies that
the probability that any coefficient in $\wh{x}'$ is nonzero is at most  $R^2 \times k/n = a^2/k = 
(a^2/k) \times (k^2/ a^2)/m = k/m$, since $m=k^2/a^2$.
This means that we can use the algorithm \textsc{BasicNoiseless2DSFFT} in subsection~\xref{subsec:basic_exact} to recover $\wh{x}'$.
Each of the entries of $\wh{x}'$ is a frequency in $\wh{x}$
which was folded into $\wh{x}'$. We employ the same phase technique used in~\cite{HIKP2} and 
subsection~\xref{subsec:basic_exact} to recover their original frequency position in $\wh{x}$.

The algorithm uses 2 functions: 
\begin{itemize}
    \item \textsc{ReduceToBasicSFFT}: This folds the spectrum into 
        $O(k)\times O(k)$ dimensions and performs the reduction to \textsc{BasicExact2DSFFT}. Note that only the $O(k)$
        elements of $x'$ which will be used in \textsc{BasicExact2DSFFT} need to be computed.
    \item \textsc{ReduceExact2DSFFT}: This invokes the reduction as well
        as the phase technique to recover $\wh{x}$.
\end{itemize}

\begin{algorithm}
\caption{Exact 2D sparse FFT algorithm for $k=o(\sqrt{n})$}
\label{a:exact}
    \begin{algorithmic}
   
    \Procedure{ReduceToBasicSFFT}{$x$, $R$, $\tau_r$, $\tau_c$}
    \State Define $x_{ij}'=x_{iR+\tau_r, jR+\tau_c}$\Comment{With lazy evaluation}
        \State \Return $\textsc{BasicExact2DSFFT}(x', k)$
    \EndProcedure

    \Procedure{ReduceExact2DSFFT}{$x$, $k$}
        \State $R \gets \frac{a\sqrt{n}}{k}$, for some constant $a < 1$ such that $R | \sqrt{n}$.
        \State $\wh{u}^{(0,0)} \gets\textsc{ReduceToBasicSFFT}(x, R, 0, 0)$
        \State $\wh{u}^{(1,0)} \gets\textsc{ReduceToBasicSFFT}(x, R, 1, 0)$
        \State $\wh{u}^{(0,1)} \gets\textsc{ReduceToBasicSFFT}(x, R, 0, 1)$

        \State $\wh{z}\gets0$
        \State $L \gets \text{supp}(\wh{u}^{(0, 0)})\cap\text{supp}(\wh{u}^{(1, 0)})\cap\text{supp}(\wh{u}^{(0, 1)})$
        \For{$(\ell,m)\in L$}
            \State $b_r\gets\wh{u}^{(1, 0)}_{\ell, m}/\wh{u}^{(0, 0)}_{\ell, m}$
            \State $i\gets\mbox{round}(\phi(b_r)\frac{\sqrt{n}}{2\pi})\ \mbox{mod}\ \sqrt{n}$
            \State $b_c\gets\wh{u}^{(0, 1)}_{\ell, m}/\wh{u}^{(0, 0)}_{\ell, m}$
            \State $j\gets\mbox{round}(\phi(b_c)\frac{\sqrt{n}}{2\pi})\ \mbox{mod}\ \sqrt{n}$
            \State $\wh{z}_{ij} \gets \wh{u}^{(0, 0)}_{\ell, m}$
        \EndFor
        \State \Return $\wh{z}$
    \EndProcedure

    \end{algorithmic}
\end{algorithm}

\paragraph{Analysis of \textsc{ReduceExact2DSFFT}}
\begin{lemma}
\label{lem:mapping}
For any constant $\alpha$, for sufficiently small $a$ there is a
one-to-one mapping of frequency coefficients from $\wh{x}$ to
$\wh{x}'$ with probability at least $1-\alpha$.
\end{lemma}

\begin{proof}
The probability that there are at least $2$ nonzero coefficients among the $R^2$ coefficients in $\wh{x}$ that are folded together in $\wh{x}'$, is at most 
\[ { R^2 \choose 2} (k/n)^2 < (a^2 n/k^2)^2 (k/n)^2 = a^4/k^2 \]
The probability that this event holds for any of the $m$ positions in $\wh{x}'$ is at most $m a^4/k^2 = (k^2/a^2)  a^4/k^2 = a^2$
which is less than $\alpha$ for small enough $a$.  Thus, with
probability at least $1 - \alpha$ any nonzero coefficient in
$\wh{x}'$ comes from only one nonzero coefficient in $\wh{x}$.
\end{proof}

\begin{theorem}\label{th:correctness_exact}
  For any constant $\alpha > 0$, there exists a constant $c > 0$ such that if  $k < c\sqrt{n}$ then the algorithm \textsc{ReduceExact2DSFFT} uses
  $O(k)$ samples, runs in time $O(k\log k)$ and returns the correct
  vector $\wh{x}$ with probablility at least $1-\alpha$.
\end{theorem}

\begin{proof}

By Theorem~\ref{th:basic} and the fact that each coefficient in $\wh{x}'$ 
is nonzero with probability $O(1/k)$, each invocation of the function \textsc{ReduceToBasicSFFT}
fails with probability at most $\alpha$. By Lemma~\ref{lem:mapping}, with probability at least $1-\alpha$, we could recover $\wh{x}$ correctly
if each of the calls to \textsc{RedToBasicSFFT} returns the correct result.
By the union bound, the algorithm \textsc{ReduceExact2DSFFT} fails with probability at most $\alpha + 3 \times \alpha  = O(\alpha)$.

The algorithm uses $O(1)$ invocations of \textsc{BasicExact2DSFFT} on a signal of size $O(k)\times O(k)$ in addition to $O(k)$ time to recover the support using the OFDM trick. 
Noting that calculating the intersection $L$ of supports takes $O(k)$ time, the stated number of samples and running time then follow directly from Theorem~\ref{th:basic}.
\end{proof}

\section{Algorithm for Exactly Sparse Case of any sparsity \lowercase{$k= \uppercase{O}(n)$}}
\subsection{Exact 1D Algorithm for $k=O(n)$}
\label{sec:1d}
We will first present a deterministic algorithm for the
one-dimensional exactly sparse case. The algorithm \textsc{1DSFFT},
described in Alg.~\ref{a:blackbox}, computes the spectrum of a
$t$-sparse signal $x \in \C^n$ and has worst case running time of
$O(t^{2} + t(\log\log{n})^{O(1)})$ for $t = O(\log n)$. This
deterministic algorithm has the fastest known worst case running time
for $k = o(\log{n})$. We will later use it to construct the algorithm
\textsc{Exact1DSFFT} that has the fastest known average case running
time for $k=O(n)$.

\begin{algorithm}
\caption{Algorithm for computing the exact 1D FFT of a $t$-sparse signal of size $n$}\label{a:blackbox}
    \begin{algorithmic}
    \Procedure{1DSFFT}{$x,t$} \label{detSFFT1D}
        \State $\wh{x} \gets 0$
	\State $\{(f_i, v_i)\}_{i \in [l]} \gets \textsc{SignalFromSyndrome}(n, \{x_{0},\cdots,x_{2t-1}\})$.
        \State  $\wh{x}_{f_{i}} \gets v_i$ for all $i \in [l]$.\Comment{$l \leq t$ is result size}
        \State \Return $\wh{x}$
    \EndProcedure\\ 

    \Procedure{SignalFromSyndrome}{$n, \{s_{0},\cdots,s_{2t-1}\}$} \label{a:cc}
            \State  $\Lambda(z) \gets \textsc{BerlekampMassey}(\{s_{0},\cdots,s_{2t-1}\})$.
            \State  $f \gets \textsc{Pan}(\Lambda(z))$. \Comment{This finds $\{f : |\wh{x}_f| > 0\}$ of length $l \leq t$}
            \State  $V \gets \textsc{Vandermonde}((\omega')^{f_0},\cdots,(\omega')^{f_{l-1}})$ 
            \State  $V^{-1} \gets$ \textsc{InverseVandermonde}($V$)
            \State  $v \gets V^{-1}s_{[l]}$
        \State \Return $\{(f_i, v_i)\}_{i \in [l]}$
    \EndProcedure
    \end{algorithmic}
\end{algorithm}

\textsc{1DSFFT} is a wrapper for the ``Signal From Syndrome''
procedure \textsc{SignalFromSyndrome} which uses the following procedures:
\begin{itemize}
\item \textsc{BerlekampMassey}: This function finds the coefficients
  of the error locator polynomial $\Lambda(z)$ based on the input
  \emph{syndromes} $s_i = \sum_i \wh{x}_i (\omega')^i$. The error
  locator polynomial, which only depends on the locations of the
  nonzero frequency components of $x$, is given by:
        \begin{equation}
            \label{equ:locator_poly}
            \Lambda(z) = \prod_{\ell=1}^{t}\left(1-{z}(\omega')^{f_\ell}\right)
        \end{equation}

Constructing the error locator polynomial is a commonly used step in decoding Reed-Solomon codes~\cite{ecc}
In our case, the main difference is that the coefficients of $\Lambda(z)$ lie in $\mathbb{C}$ whereas they lie in a finite field in the case of Reed-Solomon codes. By Lemma~\ref{lem:massey} below, the Berlekamp-Massey algorithm~\cite{massey1} solves this problem in time quadratic in $t$.
\begin{lemma}[\cite{massey1}]
\label{lem:massey}
Given the first $2t$ time-domain samples of a $t$-frequency sparse signal, Algorithm \textsc{BerlekampMassey} finds the error locator polynomial $\Lambda(z)$ (given by Equation \ref{equ:locator_poly}) in time $O(t^2)$.
\end{lemma}
\begin{proof}
  As shown in~\cite{reed-solomon}, finding the positions of the
  nonzero frequencies is equivalent to a generalization of the
  Reed-Solomon decoding problem to the complex field; then, solving
  this complex-field Reed-Solomon problem reduces to recovering the
  lowest-order linear recurrence ($\Lambda(z)$) that generates a given
  sequence of ``syndromes'' (which equal the $x_i$ in our
  case)~\cite{coding-theory}.  The running time is $O(t^2)$ where $t$
  is the degree of the polynomial ~\cite{massey1}.
\end{proof}
    \item \textsc{Pan}: By the definition of the error locator polynomial (given by Equation~\ref{equ:locator_poly}), its roots determine the set $\{f_1, \cdots, f_{t}\}$ of nonzero frequencies of $x$. Thus, we can use the Pan root-finding algorithm~\cite{pan02} to find the complex roots of $\Lambda(z)$.
\begin{lemma}
\label{thm:pan}
(\cite{pan02}) For a polynomial $p(z)$ of degree $t$ with complex coefficients and whose complex roots are located in the unit disk $\{z : |z| \le 1\}$, the \textsc{Pan} algorithm approximates all the roots of $p(z)$ with an absolute error of at most $2^{-b}$ and using $O(t\log^2{t}(\log^2{t}+\log{b}))$ arithmetic operations on $b$-bit numbers, assuming that $b  \ge t \log t$.
\end{lemma}

    \item \textsc{Vandermonde}: Given the nonzero frequencies of $x$, the problem reduces to solving a system of linear equations in the values of those
        frequencies. The coefficient matrix of this system is a Vandermonde matrix. Thus, this system has the form:
        \begin{equation}
            \label{equ:vandermonde}
            V \left[\begin{array}{c}v_1\\\vdots\\v_{t}\end{array}\right] = s^{(T)}
            \mbox{\ where  } V_{i,j} = (\omega')^{i f_j}
        \end{equation}
    \item \textsc{InverseVandermonde}: The Vandermonde matrix can be inverted using the optimal algorithm given in~\cite{vandermonde} to find the values of the nonzero coordinates of $\wh{x}$.
\begin{lemma}
\label{thm:invert_vandermonde}
(\cite{vandermonde}) Given a $t\times{t}$ Vandermonde matrix $V$, its inverse can be computed in time $O(t^2)$.
\end{lemma}
\end{itemize}

\paragraph{Analysis of \textsc{1DSFFT}}
\begin{theorem}
\label{thm:black_box}
If $x$ is a $t$-sparse signal of size $n$ where $t \le O(\log{n})$, then on input $\{x_{0},\cdots,x_{2t-1}\}$ the procedure \textsc{SignalFromSyndrome} (and hence the algorithm \textsc{1DSFFT} ) computes the spectrum $\wh{x}$ in time $O(t^{2} + t(\log\log{n})^{O(1)} )$.
\end{theorem}

\begin{proof}
  By Lemma \ref{lem:massey}, the time needed to construct the error
  locator polynomial $\Lambda(z)$ is $O(t^{2})$. By Lemma
  \ref{thm:pan} and since $t \le \log{n}$, running the \textsc{Pan}
  algorithm with $b = \log{n} \log\log{n}$ requires
  $O(t\log^2{t}(\log^2{t}+\log\log{n}))$ arithmetic operations on
  $(\log{n}\log\log{n})$-bit numbers. Since a
  $(\log{n}\log\log{n})$-bit arithmetic operation can be implemented
  using $O((\log\log{n})^{O(1)})$ $\log{n}$-bit operations and since
  we are assuming that arithmetic operations on $\log{n}$ bits can be
  performed in constant time, the total running time of the
  \textsc{Pan} algorithm is $O(t(\log\log{n})^{O(1)})$. Note that
  since the roots of $\Lambda(z)$ are $n$-th roots (or later
  $\sqrt{n}$-roots) of unity, the precision of the algorithm is
  sufficient.
  Noting that, by Lemma \ref{thm:invert_vandermonde}, inverting the Vandermonde matrix requires $O(t^{2})$ time completes the proof of the Theorem.
\end{proof}
%

\paragraph{Description and Analysis of \textsc{Exact1DSFFT}}
The algorithm \textsc{Exact1DSFFT}(Algorithm ~\ref{a:1d}) computes the
1D spectrum with high probability over a random $k$-sparse input for
any $k$. It uses $O(k)$ samples and runs in time
$O(k(\log{k}+(\log\log{n})^{O(1)}) )$ . The idea is the following: We
fold the spectrum into $\theta(k/\log{k})$ bins using the 1D version of the
comb filter (cf. Section \xref{sec:defs}). As shown in
Lemma~\ref{lem:1d_prob}, with high probability, each of the bins has
$O(\log{k})$ nonzero frequencies. In this case,
$\textsc{SignalFromSyndrome}(n,\wh{u}^{(T)}_j)$ will then recover the
original spectrum values. The generalization of this algorithm to the 2D case can be found
in Section~\ref{app:2dgen}.

\begin{algorithm}
\caption{Exact 1D sparse FFT algorithm for any sparsity $k$} \label{a:1d}
    \begin{algorithmic}
    \Procedure{Exact1DSFFT}{$x,k$}
        \State $\wh{x} \gets 0$
        \State $B \gets \Theta(k/\log{k})$ \Comment{Such that $B \mid n$}
        \State $T \gets [2C\log{k}]$ for a sufficiently large constant $C$. 
        \For{$\tau \in T$}
        \State $u^{(\tau)}_{i} := x_{i(n/B)+\tau}$ for $i \in [B]$
            \State Compute $\wh{u}^{(\tau)}$, the DFT of $u^{(\tau)}$
        \EndFor
        \For{$j \in [B]$}\Comment{$\wh{u}^{(T)}_{j}=\{\wh{u}^{(\tau)}_{j}\,:\,\tau\in{T}\}$.}
            \State $\{(v_i,f_i)\}_{i \in [C\log{k}]} \gets \textsc{SignalFromSyndrome}(n,\wh{u}^{(T)}_j)$ 
            \State  $\wh{x}_{f_{i}} \gets v_i$ for all $i \in [C\log{k}]$.
        \EndFor
        \State \Return $\wh{x}$
    \EndProcedure
    \end{algorithmic}

\end{algorithm}

\begin{lemma}
\label{lem:1d_prob}
Assume that $\wh{x}$ is distributed according to the Bernoulli model of Section \xref{sec:defs}.
If we fold the spectrum into $B=\theta(k/\log{k})$ bins, then for a sufficiently large constant $C$, the probability that there is a bin with more than $C\log{k}$ nonzero frequencies is smaller than $O((1/k)^{0.5C\log{C}})$.
\end{lemma}
\begin{proof}
The probability $\mbox{Pr}(B, k, m)$ that there is a bin with more than $m$ nonzero frequencies is bounded by:
\begin{equation*}
\label{equ:collision_prob}
\mbox{Pr}(B, k, m) \le B {n/B \choose m} \left(\frac{k}{n}\right)^{m} \le B \left( \frac{ek}{Bm} \right)^m
\end{equation*}
Since $B=d k/\log{k}$ (for some constant $d>0$), $m=C\log{k}$, we get: 
$$\mbox{Pr}(B, k, m) \le B\left( \frac{ek}{Bm} \right)^m
\le \frac{dk}{\log{k}}\left( \frac{ed}{C} \right)^{C\log{k}}
= O\left(\frac1{k^{0.5C\log{C}}}\right)$$
\end{proof}

\begin{theorem}
\label{thm:noiseless_1d_sparse_fft}
If $\wh{x}$ is distributed according to the Bernoulli model of Section \xref{sec:defs}, then Algorithm \textsc{Exact1DSFFT} runs in time $O(k(\log{k}+(\log\log{n})^{O(1)}) )$, uses $O(k)$ samples and returns the correct spectrum $\wh{x}$ with probability at least $1-O((1/k)^{0.5C\log{C}})$.
\end{theorem}
\begin{proof}
  By Lemma~\ref{lem:1d_prob}, with probability at least
  $1-O((1/k)^{0.5C\log{C}})$, all the bins have at most $C\log{k}$
  nonzero frequencies each. Then, Theorem \ref{thm:black_box}
  guarantees the success of
  $\textsc{SignalFromSyndrome}(n,\wh{u}^{(T)}_j)$ for every $j \in
  [B]$. This proves the correctness of \textsc{Exact1DSFFT}.

The running time of the for loop over $\tau$ is $O(k\log{k})$. By Theorem \ref{thm:black_box}, \textsc{SignalFromSyndrome}$(n,\wh{u}^{(T)}_j)$ takes time $O(\log^{2}{k} + \log{k}(\log\log{n})^{O(1)} )$ for every $j \in [B]$. Thus, the total running time of \textsc{Exact1DSFFT} is $O(k(\log{k}+(\log\log{n})^{O(1)}) )$. For every $\tau \in [C\log{k}]$, computing $\wh{u}^{(\tau)}$ requires $B = k/\log{k}$ samples. Thus, the total number of samples needed is $O(k)$.
\end{proof}

\subsection{Exact 2D Algorithm for \lowercase{$k=\uppercase{O}(n)$}}\label{app:2dgen}
Here, we generalize the \textsc{Exact1DSFFT} to the 2D case. 
For $k/\log{k}\ge \sqrt{n}$, the generalization is straightforward and can be found in Alg.~\ref{a:general2Dlargek}: \textsc{Exact2DSFFT1}. For $k/\log{k} \le \sqrt{n}$, the generaliztion requires an extra step and can be found in Alg.~\ref{a:2d} : \textsc{Exact2DSFFT2}. 

When $k / \log k \geq \sqrt{n}$, the desired bucket size $n/B$ is less
than $\sqrt{n}$, so we can have one-dimensional buckets and recover
the locations with a single application of syndrome decoding.  When $k
/ \log k < \sqrt{n}$, we need to have two dimensional buckets to make
them large enough.  But this means syndrome decoding will not uniquely
identify the locations, and we will need multiple tests.

\begin{algorithm}
\caption{Exact 2D sparse FFT algorithm for sparsity $k/\log{k} \geq \sqrt{n}$} \label{a:general2Dlargek}
    \begin{algorithmic}
    \Procedure{Exact2DSFFT1}{$x,k$}
        \State $B_{1} \gets \sqrt{n}$.
        \State $B_{2} \gets \theta(k/(\log{k}\sqrt{n}))$.
        \State $T \gets [2C\log{k}]$ for a sufficiently large constant $C$.
        \For{$\tau \in T$}
        \State Define $u^{(\tau)}_{i,j}: = x_{i,j(\sqrt{n}/B_{2})+\tau}$ for $(i,j) \in [B_1]\times[B_2]$.
        \State Compute  the 2D FFT $\wh{u}^{(\tau)}$ of $u^{(\tau)}$
        \EndFor

        \State $\wh{x} \gets 0$.
        \For{$(i,j) \in [B_{1}] \times [B_{2}]$}\Comment{$\wh{u}^{(T)}_{i,j}:=\{\wh{u}^{(\tau)}_{i,j}\,:\,\tau\in{T}\}$.}
	    \State $\{(f_l,v_l)\}_{l \in [C\log{k}]} \gets \textsc{SignalFromSyndrome}(\sqrt{n},\wh{u}^{(T)}_{i,j})$ 

	    \State $\wh{x}_{i,f_l} \gets v_l$ for all $l \in [C\log{k}]$.
        \EndFor
        \State \Return $\wh{x}$
    \EndProcedure
    \end{algorithmic}
\end{algorithm}

\paragraph{Description and Analysis of \textsc{Exact2DSFFT1}}
The algorithm \textsc{Exact2DSFFT1} applies to the case where \\ $k/\log{k} \geq \sqrt{n}$. As in the 1D case, we use $B=\Theta(k/\log{k})$ buckets each of which having $\Theta(n \log{k}/k)$ frequencies mapping to it. We construct the buckets corresponding to a phase shift of $\tau$ along the second dimension for all $\tau \in [2C\log k]$. As in the 1D case, with high probability, each of those buckets will have at most $C\log{k}$ nonzero frequencies. The particular choice of the buckets above will ensure that the inputs to the \textsc{SignalFromSyndrome} procedure have the appropriate ``syndrome'' form.

\begin{theorem}\label{le:analysisrunningtimelargek}
If $x$ is a $k$-sparse signal (with $k/\log{k} \geq \sqrt{n}$) distributed according to the Bernoulli model of Section \xref{sec:defs}, 
then Algorithm \textsc{Exact2DSFFT1} runs in time $O(k\log{k})$, uses $O(k)$ samples and recovers the spectrum $\wh{x}$ of $x$ with probability at least $1 - O\left(1/k^{0.5C\log{C}}\right)$.
\end{theorem}

\begin{proof}
For every $(i,j) \in [B_{1}] \times [B_{2}]$ and every $\tau \in [2 C \log{k}]$, $\wh{u}^{(\tau)}_{i,j} = \displaystyle\sum\limits_{f_{2} \equiv j \text{ mod }B_{2}} \wh{x}_{i,f_{2}} \omega^{- \tau f_{2} }$. Using the same argument as in Lemma \ref{lem:1d_prob}, with high probability, every bin $\wh{u}^{(\tau)}_{i,j}$ has at most $C\log{k}$ nonzero frequencies. Noting that the function $\textsc{SignalFromSyndrome}(\sqrt{n},\wh{u}^{(T)}_{i,j})$ succeeds whenever $\wh{u}^{(T)}_{i,j}$ are the syndromes of a $C\log{k}$-sparse signal, implies the correctness of \textsc{Exact2DSFFT1}.

Computing $\wh{u}^{(\tau)}$ for all $\tau \in T$ takes time $O(k \log{k})$. 
Each call to $\textsc{SignalFromSyndrome}(\sqrt{n},\wh{u}^{(T)}_{i,j})$ takes time $O(\log^{2}{k} + \log{k}(\log\log{n})^{O(1)})$ by Theorem \ref{thm:black_box}. 
Thus, the overall running time is $O(k(\log{k}+(\log\log{n})^{O(1)}) )=O(k\log{k})$. 
For every $\tau \in [C\log{k}]$, computing $\wh{u}^{(\tau)}$ requires $B_{1}\times B_{2} = k/\log{k}$ samples. Thus, the total number of samples needed is $O(k)$.
\end{proof}

\begin{algorithm}
\caption{Exact 2D sparse FFT algorithm for sparsity $k/\log{k} \le \sqrt{n}$} \label{a:2d}
    \begin{algorithmic}
    \Procedure{Exact2DSFFT2}{$x,k$}
        \State $B \gets \theta(k/\log{k})$.
        \State $T \gets \{0, 1, \cdots, 2C\log{k}-1\}$ for a sufficiently large constant $C$.
        \For{$(\tau,s) \in T \times [4]$}
            \State Compute  the FFT $\wh{u}^{\tau,s}$ of $u^{\tau,s}$ where for every $i \in [B]$, $u^{(\tau_{1},\tau_{2})}_{i} = x_{\tau_{1},i(\sqrt{n}/B)+\tau_{2}} $
        \EndFor

        \State $\wh{x} \gets 0$.
        \For{$i \in [B]$}

        \For{$s \in [4]$} \Comment{$\wh{u}^{T,s}_{i}=\{\wh{u}^{\tau,s}_{i}\,:\,\tau\in{T}\}$.}

	    \State $\{(f_l^{(s)}, v_l^{(s)})\}_{l \in [C\log k]} \gets \textsc{SignalFromSyndrome}(\sqrt{n},\wh{u}^{T,s}_i)$

\comments{            \State  $\Lambda(f) \gets \textsc{BerlekampMassey}(\{\wh{u}^{\tau,s}_{i}\}_{\tau})$.
            \State  $\{f_1^{s,i},\cdots,f_{C\log{k}}^{s,i}\} \gets \textsc{Pan}(\Lambda(f))$. \comments{\Comment{This finds $\{f : |\wh{x}_f| > 0\ \mbox{and}\ f = j \mod B\}$ }} 

	     \State  $V \gets \textsc{Vandermonde}(\omega^{f_1^{s,i}},\cdots,\omega^{f_{C\log{k}}^{s,i}})$ \Comment{$\omega$ is an $\sqrt{n}$th root of unity}
            \State  $V^{-1} \gets$ \textsc{InverseVandermonde}($V$) \Comment{Invert Vandermonde matrix}
            \State  $\{v_1^{s,i},\cdots,v_{C\log{k}}^{s,i}\} \gets V^{-1}\wh{u}^{T,s}_{i}$. \Comment{$\wh{u}_{i}^{T,s}=\{\wh{u}_{i}^{\tau,s}\,:\,\tau\in{T}\}$}
}
        \EndFor

	\State $(\{f_0,\cdots,f_{O(\log{k})}\}, \{y_0^{(s)},\cdots,y_{O(\log{k})}^{(s)}\}_{s \in [4]}) \gets \textsc{Match}(\{(f_l^{(s)},v_l^{(s)})\}_{s \in [4], ~  l \in [C\log{k}]})$

        \For{$l \in [O(\log{k})]$}\Comment{$y_l^{(S)}:=\{y_l^{(s)}\,:\,s \in S=[4]\}$.}
	    \State $\{(g_0,w_0),(g_1, w_1)\} \gets \textsc{SignalFromSyndrome}(\sqrt{n},y_l^{(S)})$ 
	    \State $\wh{x}_{f_l,g_j} \gets w_j$ for all $j \in [2]$.
        \EndFor
        \EndFor

        \State \Return $\wh{x}$
    \EndProcedure
    \end{algorithmic}
\end{algorithm}

\paragraph{Description and Analysis of \textsc{Exact2DSFFT2}}
The algorithm \textsc{Exact2DSFFT2} above applies to the case where
$k/\log{k} \le \sqrt{n}$. As in the $1$D case, we use
$B=\Theta(k/\log{k})$ buckets, each of which having $\Theta(n
\log{k}/k)$ frequencies mapping to it (i.e. $\Theta(\sqrt{n}
\log{k}/k)$ columns). We construct $4$ sets of buckets (as opposed to
$1$ set in the 1D case). Those sets correspond to the phase shifts
$(\tau,0)$, $(\tau,1)$, $(\tau,2)$ and $(\tau,3)$ for all $\tau \in
[2C\log k]$. We run the \textsc{SignalFromSyndrome} procedure on each
of those $4$ sets. As opposed to the 1D case, the resulting values can
be the superposition of $2$ or more nonzero frequency
components. However, as shown in Lemma \ref{le:collisionthree}, with
high probability, all the obtained values correspond to the
superposition of at most $2$ nonzero frequency components. The $4$
corresponding superpositions (one from each of the $4$ sets) are then
combined (by the \textsc{Match} procedure) to get the union
$\{f_0,\cdots,f_{O(\log{k})}\}$ of the sets
$\{f_0^{(s)},\cdots,f_{C\log{k}-1}^{(s)}\}$ for all $s \in [4]$ along
with the associated values $\{y_0^{(s)},\cdots,y_{O(\log{k})}^{(s)}\}$
(with a value $0$ if the frequency did not appear for some $s$). Then,
we give the $4$ resulting superpositions as inputs to the
\textsc{SignalFromSyndrome} procedure again. The particular choice of
the $4$ sets of buckets above ensures that those inputs have the
appropriate ``syndrome'' form. The output of this procedure will then
consist of original spectrum values.

\begin{lemma}\label{le:blackboxfirst}
  With probability at least $1 - O\left(1/k^{0.5C\log{C}}\right)$, for
  every $i \in [B]$ and $s \in [4]$, the output of
  $\textsc{SignalFromSyndrome}(\sqrt{n},\wh{u}^{T,s}_i)$ consists of
  all nonzero values of the form $\displaystyle\sum\limits_{f_{2}
    \equiv i \text{ mod }B} \wh{x}_{f_{1},f_{2}} \omega^{- s f_{2} }$
  for some $f_{1} \in [\sqrt{n}]$.
\end{lemma}

\begin{proof}
  As in Lemma \ref{lem:1d_prob}, we have that the probability that
  there is a bin with more than $\abs{T}/2 = C\log{k}$ nonzero frequencies is at
  most $O\left(1/k^{0.5C\log{C}}\right)$. Moreover, for every $\tau
  \in T$, $i \in [B]$ and $s \in [4]$, we have:
\begin{align}
\label{equ:superposition_vandermode}
\wh{u}^{\tau,s}_i & = &  \displaystyle\sum\limits_{f_{1} \in [\sqrt{n}]} \displaystyle\sum\limits_{f_{2} \equiv i \text{ mod }B} \wh{x}_{f_{1},f_{2}} \omega^{-\tau f_{1} - s f_{2}} \\
& = & \displaystyle\sum\limits_{f_{1} \in [\sqrt{n}]} \big(\displaystyle\sum\limits_{f_{2} \equiv i \text{ mod }B} \wh{x}_{f_{1},f_{2}} \omega^{- s f_{2} } \big) \omega^{-\tau f_{1}}
\end{align}
Noting that the function
$\textsc{SignalFromSyndrome}(\sqrt{n},\wh{u}^{T,s}_i)$ succeeds
whenever $\wh{u}^{T,s}_i$ are the syndromes of a $\abs{T}/2$-sparse
signal, we get the desired statement.
\end{proof}

\begin{lemma}\label{le:collisionthree}
The probability that there are more than $2$ nonzero frequency components that superimpose in a power of $\omega$ 
(i.e., as in Equation~(\ref{equ:superposition_vandermode}), $\wh{x}_{f_1, f_2}$ and $\wh{x}_{f_1', f_2'}$ superimpose
if $f_1 = f'_1$ and $f_2\equiv f_2'\ \mbox{mod}\ B$)
is at most $O(\frac{k \log^{2}{k}}{n})$.
\end{lemma}

\begin{proof}
Since $\sqrt{n}/B = \Theta(\sqrt{n}\log{k}/k)$ frequencies map to each power of $\omega$ in each bucket, the probability is upper bounded by 
\[
n \cdot k/n \cdot {\sqrt{n} / B \choose 2} \left
  (\frac{k}{n}\right)^{2} \le \frac{k^3}{2nB^2} = \Theta(\frac{k \log^{2}{k}}{n})
\]
\end{proof}

\begin{lemma}\label{le:blackboxsecond}
  With probability at least \[ 1 - O\left(k \log^{2}{k}/n -
  1/k^{0.5C\log{C}}\right), \] for all $i \in [B]$ the outputs
  of $\textsc{SignalFromSyndrome}(\sqrt{n},y_l^{S})$ for all $l \in [C \log{k}]$
  consist of all nonzero $\wh{x}_{f_{1},f_{2}}$ where $f_{2} \equiv i
  \text{ mod }B$.
\end{lemma}

\begin{proof}
  By Lemmas \ref{le:collisionthree} and \ref{le:blackboxfirst}, the
  probability that all bins have at most $C\log{k}$ nonzero
  frequencies and all powers of $\omega$ have at most $2$ nonzero
  frequencies is at least $1 - O\left(k \log^{2}{k}/n -
  1/k^{0.5C\log{C}}\right)$. Then for every $i \in [B]$ and $s \in
  S=[4]$, there are at most $C\log{k}$ nonzero values of the form
  $\displaystyle\sum\limits_{f_{2} \equiv i \text{ mod }B}
  \wh{x}_{f_{1},f_{2}} \omega^{- s f_{2} }$ where $f_{1} \in [\sqrt{n}]$,
  and each of those sums consists of at most $2$ terms. Thus,
  $y_l^{(S)}$ are the syndromes of a $2$-sparse signal of the form
  $\wh{x}_{f_l,g_0} \omega^{-g_0 r} + \wh{x}_{f_l,g_1}
  \omega^{-g_1 r}$ where $r$ is the time-domain index. This yields the
  desired statement.
\end{proof}

\begin{theorem}\label{le:analysisrunningtime}
  If $x$ is a $k$-sparse signal (with $k/\log{k} \le \sqrt{n}$)
  distributed according to the Bernoulli model of Section
  \xref{sec:defs}, then Algorithm \textsc{Exact2DSFFT2} runs in time
  $O(k(\log{k}+(\log\log{n})^{O(1)}) )$, uses $O(k)$ samples and
  recovers the spectrum $\wh{x}$ of $x$ with probability at least $1 -
  k \log^{2}{k}/n - O\left(1/k^{0.5C\log{C}}\right)$.
\end{theorem}

\begin{proof}
  Lemma \ref{le:blackboxsecond} implies that Algorithm
  \textsc{Exact2DSFFT2} succeeds with the desired probability.

  Computing $\wh{u}^{\tau,s}$ for all $\tau \in T$ and all $s \in [4]$
  takes time $O(k \log{k})$. The running time of the \textsc{Match}
  procedure is $O(\log k)$. By Theorem \ref{thm:black_box}, each call
  to \textsc{SignalFromSyndrome} in the for loop over $s$ takes time
  $O(\log^{2}{k} + \log{k}(\log\log{n})^{O(1)})$ whereas each one in
  the for loop over $l$ takes time $O((\log\log{n})^{O(1)})$. Thus,
  the overall running time is $O(k(\log{k}+(\log\log{n})^{O(1)}) )$.

  For every $\tau \in [C\log{k}]$, computing $\wh{u}^{(\tau)}$
  requires $B = k/\log{k}$ samples. Thus, the total number of samples
  needed is $O(k)$.
\end{proof}


\section{Algorithm for Robust Recovery}\label{sec:general}

\begin{algorithm}[ht!]
 \caption{Robust 2D sparse FFT algorithm for $k = \Theta(\sqrt{n})$}\label{a:robust2}
  \begin{algorithmic}
    \Procedure{RobustEstimateCol}{$\wh{u}$, $\wh{v}$, $T$, $T'$, IsCol, $J$, Ranks}
    \State $\wh{w} \gets 0$.
    \State $S \gets \{\}$ \Comment{Set of changes, to be tested next round.}
    \For{$j \in J$}
    \State \textbf{continue} if $\text{Ranks}[(\mbox{IsCol}, j)] \geq \log \log n$.
    \State $i \gets \textsc{HIKPLocateSignal}(\wh{u}^{(T')}, T')$
    \Comment{Procedure from~\cite{HIKP2}: $O(\log^2 n)$ time}
    \State $a \gets \median_{\tau \in T} \wh{u}^{\tau}_{j} \omega^{\tau i}$.
    \State \textbf{continue} if $\abs{a} < L/2$ \Comment{Nothing significant recovered}
    \State \textbf{continue} if $ \sum_{\tau \in T} \sabs{\wh{u}^{\tau}_{j} - a \omega^{-\tau i}}^2 \geq L^2\abs{T}/10$
    \Comment{Bad recovery: probably not 1-sparse}
    \State $b \gets \mean_{\tau \in T} \wh{u}^{\tau}_{j} \omega^{\tau i}$.
    \If{IsCol}  \Comment {whether decoding column or row}
    \State $\wh{w}_{i,j} \gets b$.
    \Else
    \State $\wh{w}_{j,i} \gets b$.
    \EndIf
    \State $S \gets S \cup \{i\}$.
    \State $\text{Ranks}[(1-\text{IsCol}, i)] \text{ += } \text{Ranks}[(\text{IsCol}, j)]$.
    \For{$\tau\in{T \cup T'}$}
    \State $\wh{u}^{(\tau)}_j \gets \wh{u}^{(\tau)}_j - b\omega^{-\tau i}$ 
    \State $\wh{v}^{(\tau)}_i \gets \wh{v}^{(\tau)}_i -  b\omega^{-\tau i}$
    \EndFor
    \EndFor
    \State \Return $\wh{w}$, $\wh{u}$, $\wh{v}$, $S$
    \EndProcedure
    \Procedure{Robust2DSFFT}{$x$, $k$}
    \State $T,T' \subset [\sqrt{n}], \abs{T} = \abs{T'}=O(\log n)$ 
    \For{$\tau \in T \cup T'$}
    \State $\wh{u}^{(\tau)} \gets \textsc{FoldToBins}(x, \sqrt{n}, 1, 0, \tau)$.
    \State $\wh{v}^{(\tau)} \gets \textsc{FoldToBins}(x, 1, \sqrt{n}, \tau, 0)$.
    \EndFor
    \State $\wh{z} \gets 0$
    \State $\text{Ranks} \gets 1^{[2] \times [\sqrt{n}]}$ \Comment{Rank of vertex (iscolumn, index)}
    \State $S_{col} \gets [\sqrt{n}]$ \Comment{Which columns to test}
    \For{$t \in [C\log n]$}
    \State $\{\wh{w}, \wh{u}, \wh{v}, S_{row}\} \gets \textsc{RobustEstimateCol}(\wh{u}, \wh{v}, T,T', $ true, $S_{col}$, Ranks). 
    \State $\wh{z} \gets \wh{z}  + \wh{w}$.
    \State $S_{row} \gets [\sqrt{n}]$ if $t = 0$ \Comment{Try every row the first time}
    \State $\{\wh{w}, \wh{v}, \wh{u}, S_{col}\}  \gets \textsc{RobustEstimateCol}(\wh{v}, \wh{u}, T, T'$ false, $S_{row}$, Ranks).
    \State $\wh{z} \gets \wh{z}  + \wh{w}$.
    \EndFor
    \State \Return $\wh{z}$
    \EndProcedure
  \end{algorithmic}
\end{algorithm}

\subsection{Preliminaries}\label{sec:agenprelim}
Following~\cite{CTao} we say that a matrix $A$ satisfies a {\em restricted isometry property} (RIP) of order $t$ with constant $\delta >0$ if, for all $t$-sparse vectors $y$, we have
$\norm{2}{Ay}^2 /
  \norm{2}{y}^2 \in [1-\delta, 1+\delta]$. 
  
Suppose all columns $A_i$ of an $N \times M$ matrix $A$ have unit norm. Let $\mu=\max_{i \neq j} \abs{A_i \cdot A_j}$ be the {\em coherence} of $A$. It is folklore\footnote{It is a direct corollary of Gershgorin's theorem applied to any $t$ columns of $A$.}  that $A$ satisfies the RIP of order $t$  with the constant 
$\delta=(t-1)\mu$.

Suppose that the matrix $A$ is an $M \times N$ submatrix of the $N
\times N$ Fourier matrix $F$, with each the $M$ rows of $A$ chosen
uniformly at random from the rows of $F$.  It is immediate from the
Hoeffding bound that if $M = b \mu^2 \log (N/\gamma)$ for some large
enough constant $b>1$ then the matrix $A$ has coherence at most $\mu$
with probability $1-\gamma$. Thus, for $M= \Theta(t^2 \cdot t \log
N)$, $A$ satisfies the RIP of order $t$ with constant $\delta=0.5$
with probability $1-1/N^t$.


The algorithm appears in Algorithm~\ref{a:robust2}.

\subsection{Correctness of each stage of recovery}
\newcommand{\yhead}{y^{head}}
\newcommand{\ygauss}{y^{gauss}}
\newcommand{\yres}{y^{residue}}

\begin{lemma}\label{L:FAILUREPROB}
  Consider the recovery of a column/row $j$ in
  \textsc{RobustEstimateCol}, where $\wh{u}$ and $\wh{v}$ are the
  results of \textsc{FoldToBins} on $\wh{x}$.  Let $y \in
  \C^{\sqrt{n}}$ denote the $j$th column/row of $\wh{x}$.  Suppose $y$
  is drawn from a permutation invariant distribution $y = \yhead +
  \yres + \ygauss$, where $\min_{i \in \supp(\yhead)} \abs{y_i} \geq
  L$, $\norm{1}{\yres} < \eps L$, and $\ygauss$ is drawn from the
  $\sqrt{n}$-dimensional normal distribution $N_{\C}(0,
  \sigma^2I_{\sqrt{n}})$ with standard deviation $\sigma = \eps L /
  n^{1/4}$ in each coordinate on both real and imaginary axes.  We do
  not require that $\yhead$, $\yres$, and $\ygauss$ are independent
  except for the permutation invariance of their sum.

    Consider the following bad events:
    \begin{itemize}
    \item False negative: $\supp(\yhead)=\{i\}$ and
      \textsc{RobustEstimateCol} does not update coordinate $i$.
    \item False positive: \textsc{RobustEstimateCol} updates some
      coordinate $i$ but $\supp(\yhead)\neq \{i\}$.
    \item Bad update: $\supp(\yhead)=\{i\}$ and coordinate $i$ is
      estimated by $b$ with $\abs{b - \yhead_i} > \norm{1}{\yres} +
      \sqrt{\frac{\log \log n}{\log n}}\eps L$.
    \end{itemize}

    For any constant $c$ and $\eps$ below a sufficiently small
    constant, there exists a distribution over sets $T,T'$ of size
    $O(\log n)$, such that as a distribution over $y$ and $T,T'$ we
    have
    \begin{itemize}
    \item The probability of a false negative is $1/\log^c n$.
    \item The probability of a false positive is $1/n^c$.
    \item The probability of a bad update is $1/\log^c n$.
    \end{itemize}
\end{lemma}


\begin{proof}
  Let $\check{y}$ denote the 1-dimensional inverse DFT of $y$.  Note that
  \[
  \wh{u}^{(\tau)}_j = \check{y}_\tau
  \]
  by definition.  Therefore, the goal of \textsc{RobustEstimateCol} is
  simply to perform reliable $1$-sparse recovery with $O(\log n)$
  queries. Fortunately,~\cite{HIKP2} solved basically the same
  problem, although with more false positives than we want here.

  We choose $T'$ according to the \textsc{LocateInner} procedure
  from \cite{HIKP2}; the set $T$ is chosen uniformly at random from
  $[\sqrt{n}]$.  We have that
  \[
  \wh{u}^{(\tau)}_j = \sum_{i \in [\sqrt{n}]} y_i \omega^{-\tau i}.
  \]
  This is exactly what the procedure~\textsc{HashToBins}
  of~\cite{HIKP2} approximates up to a small error term.  Therefore,
  the same analysis goes through (Lemma~4.5 of~\cite{HIKP2}) to get
  that \textsc{HIKPLocateSignal} returns $i$ with $1 - 1/\log^c n$
  probability if $\abs{y_i} \geq \norm{2}{y_{-i}}$, where we define
  $y_{-i} := y_{[\sqrt{n}] \setminus \{i\}}$.

  Define $A \in \C^{\abs{T}\times \sqrt{n}}$ to be the rows of the
  inverse Fourier matrix indexed by $T$, normalized so $\abs{A_{i,j}} = 1$.
  Then $\wh{u}^{(\tau)}_j = (Ay)_{\tau}$.

  First, we prove
  \begin{align}\label{eq:ynorm}
    \norm{2}{\yres + \ygauss} = O(\eps L)
  \end{align}
  with all but $n^{-c}$ probability.  We have
  that $\E[\norm{2}{\ygauss}^2] = 2\eps^2 L^2$, so $\norm{2}{\ygauss}
  \leq 3\eps L$ with all but $e^{-\Omega(\sqrt{n})}  <1/n^c$
  probability by concentration of chi-square variables.  We also have
  that $\norm{2}{\yres} \leq \norm{1}{\yres} \leq \eps L$.

  Next, we show
  \begin{align}\label{eq:Aynorm}
    \norm{2}{A(\yres + \ygauss)} = O(\eps L\sqrt{|T|})
  \end{align}
  with all but $n^{-c}$ probability.  We have that $A\ygauss$ is drawn
  from $N_{\C}(0, \eps^2L^2I_{\abs{T}})$ by the rotation invariance of
  Gaussians, so
  \begin{align}\label{eq:Aygauss}
  \norm{2}{A\ygauss} \leq 3\eps L \sqrt{\abs{T}}
  \end{align}
   with all but
  $e^{-\Omega(\abs{T})} < n^{-c}$ probability.  Furthermore, $A$ has
  entries of magnitude $1$ so $\norm{2}{A\yres} \leq
  \norm{1}{\yres}\sqrt{|T|} = \eps L \sqrt{|T|}$.

  Consider the case where $\supp(\yhead) = \{i\}$.  From
  Equation~\eqref{eq:ynorm} we have
  \begin{align}
  \label{eq:ygauss}
 \norm{2}{y_{-i}}^2 \le \norm{2}{\ygauss + \yres}^2 \leq O(\eps^2L^2) < L^2
  \leq \norm{2}{y_i}^2
  \end{align}
  so $i$ is located with $1 - 1/\log^c n$ probability by
  \textsc{HIKPLocateSignal}.

  Next, we note that for any $i$, as a distribution over $\tau\in [\sqrt{n}]$,
  \[
  \E_\tau[ \abs{\wh{u}^{(\tau)}_j - y_i\omega^{-\tau i}}^2] = \norm{2}{y_{-i}}^2
  \]
  and so (analogously to Lemma~4.6 of~\cite{HIKP2}, and for any $i$),
  since $a = \median_{\tau \in T} \wh{u}^{(\tau)}_j\omega^{\tau i}$ we
  have
  \begin{equation}
  \label{eq:ay}
  \sabs{a - y_i}^2 \leq 5\norm{2}{y_{-i}}^2
  \end{equation}
  with probability $1 - e^{-\Omega(\abs{T})} = 1 - 1/n^c$ for some
  constant $c$.  Hence if $\{i\} = \supp(\yhead)$, we have $\sabs{a - y_i}^2
  \leq O(\eps^2L^2)$ and therefore $\abs{a} > L/2$, passing the first check
  on whether $i$ is valid.

  For the other check, we have that with $1 - 1/n^c$ probability
  \begin{align*}
   (\sum_{\tau \in T} \abs{\wh{u}^{(\tau)}_j - a\omega^{-\tau
        i}}^2)^{1/2}
    &= \norm{2}{A(y - ae_i)}\\
    &\leq \norm{2}{A(\ygauss + \yres + (\yhead_i - a)e_i)}\\
    &\leq \norm{2}{A(\ygauss+\yres)} + \abs{\yhead_i - a}\sqrt{\abs{T}}\\
    &\leq \norm{2}{A(\ygauss+\yres)} + (\abs{\yres_i+\ygauss_i} + \abs{y_i - a})\sqrt{\abs{T}}\\
   &\leq O(\eps L\sqrt{\abs{T}}).
  \end{align*}
  where the last step uses Equation \ref{eq:Aynorm}.
  This gives
  \[
  \sum_{\tau \in T} \abs{\wh{u}^{(\tau)}_j - a\omega^{-\tau i}}^2 = O(\eps^2 L^2 \abs{T}) < L^2 \abs{T} / 10
  \]
  so the true coordinate $i$ passes both checks.  Hence the
  probability of a false negative is $1/\log^c n$ as desired.

  Now we bound the probability of a false positive.  First consider
  what happens to any other coordinate $i' \neq i$ when
  $\abs{\supp(\yhead)} = \{i\}$.  We get some estimate $a'$ of its
  value.  Since $A/\sqrt{\abs{T}}$ satisfies an RIP of order 2 and
  constant $1/4$, by the triangle inequality and
  Equation~\ref{eq:Aynorm} we have that with $1 - n^{-c}$ probability,
  \begin{align*}
     \norm{2}{A(y - a'e_{i'})} 
   &\geq   \norm{2}{A(\yhead_ie_i - a'e_{i'})} - \norm{2}{A(\ygauss + \yres)}\\
   &\geq    \yhead_i \sqrt{\abs{T}}\cdot (3/4) - O(\eps L \sqrt{\abs{T}})\\
   &>   L\sqrt{|T|}/2.
  \end{align*}
  Hence the second condition will be
  violated, and $i'$ will not pass.  Thus if $\abs{\supp(\yhead)} = 1$,
  the probability of a false positive is at most $n^{-c}$.

  Next, consider what happens to the result $i$ of
  \textsc{HIKPLocateSignal} when $\abs{\supp(\yhead)} = 0$.  From
  Equation~\eqref{eq:ynorm} and Equation~\eqref{eq:Aynorm} we have
  that with $1 - n^{-c}$ probability:
  \[
  \sabs{a - y_i}^2 \leq 5 \norm{2}{y_{-i}}^2 \leq O(\eps^2 L^2).
  \]
  Therefore, from Equation~\ref{eq:ynorm},
  \[
  \abs{a} \leq \sabs{y_i} + \sabs{a - y_i} \leq  \norm{2}{\yres + \ygauss} +  \sabs{a - y_i}  =O(\eps L) < L/2
  \]
  so the first check is not passed and $i$ is not recovered.

  Now suppose $\abs{\supp(y_{head})} > 1$.
  Lemma~\ref{l:testrobust} says that with $1 - n^{-c}$ probability over the
  permutation, no $(i, a)$ satisfies
  \[
  \norm{2}{A(\yhead - ae_{i})}^2 < L^2 \abs{T}/5.
  \]
  But then, from Equation~\ref{eq:Aygauss}
  \begin{align*}
   \norm{2}{A(y - ae_{i})} &\geq \norm{2}{A(\yhead - ae_{i})} - \norm{2}{A\ygauss} \\
  &> L \sqrt{\abs{T}/5} - O(\eps L \sqrt{\abs{T}}) \\
  &> L \sqrt{\abs{T}/10}
  \end{align*}
  so no $i$ will pass the second check.  Thus the probability of a
  false positive is $1/n^c$.

  Finally, consider the probability of a bad update.  We have
  that
  \[
  b = \mean_{\tau \in T} (Ay)_\tau \omega^{\tau i} = \yhead_i + \mean_{\tau \in T} (A\yres + A\ygauss)_\tau \omega^{\tau i}
  \]
  and so
  \[
  \abs{b - \yhead_i} \leq \abs{\mean_{\tau \in T} (A\yres)_\tau \omega^{\tau
      i}} + \abs{\mean_{\tau \in T} (A\ygauss)_\tau \omega^{\tau i}}.
  \]
  We have that \[ \abs{\mean_{\tau \in T} (A\yres)_\tau \omega^{\tau i}} \leq
  \max_{\tau\in T} \abs{(A\yres)_{\tau}} \leq \norm{1}{\yres}\].

  We know that $A\ygauss$ is $N_{\C}(0, \eps^2 L^2 I_{\abs{T}})$.
  Hence its mean is a complex Gaussian with standard deviation $\eps L
  / \sqrt{\abs{T}}$ in both the real and imaginary axes.  This means
  the probability that
  \[
  \abs{b - \yhead_i} > \norm{1}{\yres} + t \eps L / \sqrt{\abs{T}}
  \]
  is at most $e^{-\Omega(t^2)}$.  Setting $t = \sqrt{\log \log^c n}$ gives a
  $1/\log^c n$ chance of a bad update, for sufficiently large $\abs{T}
  = O(\log n)$.
\end{proof}

The following is the robust analog of Lemma~\ref{l:test}.
\begin{lemma}\label{l:testrobust}
  Let $y \in \C^m$ be drawn from a permutation invariant distribution
  with $r \geq 2$ nonzero values.  Suppose that all the nonzero
  entries of $y$ have absolute value at least $L$. Choose $T \subset
  [m]$ uniformly at random with $t := \abs{T} = O(c^3 \log m)$

  Then, the probability that there exists a $y'$ with $\norm{0}{y'}
  \leq 1$ and
  \[
  \norm{2}{(\check{y}-\check{y}')_T}^2 <  \eps L^2 t/n
  \]
  is at most $c^3 (\frac{c}{m-r})^{c-2}$ whenever $\eps < 1/8$.
\end{lemma}

\begin{proof}
  Let $A = \sqrt{1/t} F_{T \times *}$ be $\sqrt{1/t}$ times the
  submatrix of the Fourier matrix with rows from $T$, so
  \[
  \norm{2}{(\check{y}-\check{y}')_T}^2 = \norm{2}{A(y-y')}^2 t/n.
  \]
  By a coherence bound (see Section~\ref{sec:agenprelim}), with $1 -
  1/m^c$ probability $A$ satisfies the RIP of order $2c$ with constant
  $0.5$. We would like to bound
  \[
  P := \Pr[\exists y' : \norm{2}{A(y - y')}^2 < \eps L^2 \andop \norm{0}{y'} \leq 1]
  \]

  If $r \leq c-1$, then $y - y'$ is $c$-sparse and
  \begin{align*}
    \norm{2}{A(y - y')}^2 &\geq \norm{2}{y-y'}^2/2\\
    &\geq (r-1)L^2 /2\\
    &> \eps L^2
  \end{align*}
  as long as $\eps < 1/2$, giving $P = 0$.  Henceforth, we can assume $r
  \geq c$. When drawing $y$, first place $r-(c-1)$ coordinates into $u$ then
  place the other $c-1$ values into $v$, so that $y = u + v$.
  Condition on $u$, so $v$ is a permutation distribution over
  $m-r+c-1$ coordinates.  We would like to bound
  \[
  P = \Pr_v[\exists y' : \norm{2}{A(u + v - y')}^2 <
  \eps L^2  \andop \norm{0}{y'} \leq 1].
  \]

  Let $w$ be any $c$-sparse vector such that $\norm{2}{A(u + w)}^2 <
  \eps L^2 $ (and note that if no such $w$ exists, then since $v-y'$
  is $c$-sparse, $P = 0$).  Then recalling that for any norm
  $\norm{}{\cdot}$, $\norm{}{a}^2 \leq 2\norm{}{b}^2+2\norm{}{a+b}^2$
  and hence $\norm{}{a+b}^2 \geq \norm{}{a}^2/2 - \norm{}{b}^2$,
  \begin{align*}
    \norm{2}{A(u + v - y')}^2 &\geq \norm{2}{A(v - y' - w)}^2/2 - \norm{2}{A(u + w)}^2\\
    &\geq  \norm{2}{v - y' + w}^2/4   -  \eps L^2.
  \end{align*}
  Hence
  \[
  P \leq \Pr_v[\exists y' : \norm{2}{v - y' + w}^2 <
  8\eps L^2  \andop \norm{0}{y'} \leq 1].
  \]
  Furthermore, we know that $\norm{2}{v - y' + w}^2 \geq
  L^2(\abs{\supp(v) \setminus \supp(w)} - 1)$.  Thus if $\eps < 1/8$,
  \begin{align*}
    P &\leq \Pr_v[\abs{\supp(v) \setminus \supp(w)} \leq 1]\\
    &\leq \frac{c + (m-r+c-1)c(c-1)/2}{\binom{m-r+c-1}{c-1}}\\ 
    &<  c^3 (\frac{c}{m-r})^{c-2}
  \end{align*}
  as desired.
\end{proof}

\subsection{Overall Recovery}

Recall that we are considering the recovery of a signal $\wh{x} =
\wh{x^*} + \wh{w} \in \C^{\sqrt{n}\times\sqrt{n}}$, where $\wh{x^*}$
is drawn from the Bernoulli model with expected $k = a\sqrt{n}$
nonzeros for a sufficiently small constant $a$, and $\wh{w} \sim
N_\C(0, \sigma^2 I_n)$ with $\sigma = \eps L \sqrt{k/n} = \Theta(\eps
L / n^{1/4})$ for sufficiently small $\eps$.

It will be useful to consider a bipartite graph
representation $G$ of $\wh{x^*}$.  We construct a bipartite graph with
$\sqrt{n}$ nodes on each side, where the left side corresponds to rows
and the right side corresponds to columns.  For each $(i, j) \in
\supp(\wh{x^*})$, we place an edge between left node $i$ and right
node $j$ of weight $\wh{x^*}_{(i, j)}$.

Our algorithm is a ``peeling'' procedure on this graph.  It iterates
over the vertices, and can with a ``good probability'' recover an edge
if it is the only incident edge on a vertex.  Once the algorithm
recovers an edge, it can remove it from the graph.  The algorithm will
look at the column vertices, then the row vertices, then repeat; these
are referred to as \emph{stages}.  Supposing that the algorithm
succeeds at recovery on each vertex, this gives a canonical order to
the removal of edges.  Call this the \emph{ideal} ordering.

In the ideal ordering, an edge $e$ is removed based on one of its
incident vertices $v$.  This happens after all other edges reachable
from $v$ without passing through $e$ are removed.  Define the
\emph{rank} of $v$ to be the number of such reachable edges, and
$\text{rank}(e) = \text{rank}(v) + 1$ (with $\text{rank}(v)$ undefined
if $v$ is not used for recovery of any edge).

\begin{lemma}\label{l:ranks}
  Let $c, \alpha$ be arbitrary constants, and $a$ a sufficiently small
  constant depending on $c, \alpha$.  Then with $1 - \alpha$
  probability every component in $G$ is a tree and at most $k/\log^c
  n$ edges have rank at least $\log \log n$.
\end{lemma}
\begin{proof}
  Each edge of $G$ appears independently with probability $k/n =
  a/\sqrt{n}$.  There are at most $\sqrt{n}^t$ cycles of length $t$.
  Hence the probability that any cycle of length $t$ exists is at most
  $a^t$, so the chance any cycle exists is less than $a^2/(1-a^2) <
  \alpha/2$ for sufficiently small $a$.

  Each vertex has expected degree $a < 1$. Exploring the component for
  any vertex $v$ is then a subcritical branching process, so the
  probability that $v$'s component has size at least $\log \log n$ is
  $1/\log^c n$ for sufficiently small $a$.  Then for each edge, we
  know that removing it causes each of its two incident vertices to
  have component size less than $\log \log n - 1$ with $1 - 1/\log^c
  n$ probability.  Since the rank is one more than the size of one of
  these components, the rank is less than $\log \log n$ with $1 -
  2/\log^c n$ probability.

  Therefore, the expected number of edges with rank at least $\log \log n$
  is $2k/\log^c n$.  Hence with probability $1 - \alpha/2$ there are
  at most $(1/\alpha)4k/\log^c n$ such edges; adjusting $c$ gives the
  result.
\end{proof}
\begin{lemma}
\label{l:robust}
  Let \textsc{Robust2DSFFT'} be a modified
   \textsc{Robust2DSFFT} that avoids false negatives or bad
  updates: whenever a false negative or bad update would occur, an
  oracle corrects the algorithm.  With large constant probability,
  \textsc{Robust2DSFFT'} recovers $\wh{z}$ such that there exists
  a $(k/\log^c n)$-sparse $\wh{z}'$ satisfying
  \[
  \norm{2}{\wh{z}-\wh{x} - \wh{z}'}^2 \leq 6\sigma^2 n.
  \]
  Furthermore, only $O(k/\log^c n)$ false positives or bad updates are
  caught by the oracle.
\end{lemma}
\begin{proof}
  One can choose the random $\wh{x^{*}}$ by first selecting the
  topology of the graph $G$, and then selecting the random ordering of
  the columns and rows of the matrix. Note that reordering the
  vertices only affects the ideal ordering by a permutation within
  each stage of recovery; the set of edges recovered at each stage in
  the ideal ordering depends only on the topology of $G$.  Suppose
  that the choice of the topology of the graph satisfies the thesis of
  Lemma~\ref{l:ranks} (which occurs with large constant
  probability). We will show that with large constant probability
  (over the space of random permutations of the rows and
  columns),\xxx{PI:A little convoluted, but hopefully understandable}
  \textsc{Robust2DSFFT'} follows the ideal ordering and the
  requirements of Lemma~\ref{L:FAILUREPROB} are satisfied at every
  stage.

  For a recovered edge $e$, we define the ``residue'' $\wh{x^*}_e -
  \wh{z}_e$.  We will show that if $e$ has rank $r$, then
  $\abs{\wh{x^*}_e - \wh{z}_e} \leq r \sqrt{\frac{\log \log n}{\log n}}\eps L$.

  During attempted recovery at any vertex $v$ during the ideal
  ordering (including attempts on vertices which do not have exactly
  one incident edge), let $y \in \C^{\sqrt{n}}$ be the associated
  column/row of $\wh{x}-\wh{z}$.  We split $y$ into three parts $y =
  \yhead + \yres + \ygauss$, where $\yhead$ contains the elements of
  $\wh{x^*}$ not in $\supp(\wh{z})$, $\yres$ contains
  $\wh{x^*}-\wh{z}$ over the support of $\wh{z}$, and $\ygauss$
  contains $\wh{w}$ (all restricted to the column/row corresponding to
  $v$).  Let $S = \supp(\yres)$ contain the set of edges incident on
  $v$ that have been recovered so far.  We have by the inductive
  hypothesis that $\norm{1}{\yres} \leq \sum_{e \in S} \text{rank}(e)
  \sqrt{\frac{\log \log n}{\log n}}\eps L$.  Since the algorithm
  verifies that $\sum_{e \in S} \text{rank}(e) \leq \log \log n$, we
  have
  \[
  \norm{1}{\yres} \leq \sqrt{\frac{\log^3 \log n}{\log n}}\eps L < \eps L.
  \]

  Furthermore, $y$ is permutation invariant: if we condition on the
  values and permute the rows and columns of the matrix, the algorithm
  will consider the permuted $y$ in the same stage of the algorithm.

  Therefore the conditions for Lemma~\ref{L:FAILUREPROB} hold.  This
  means that the chance of a false positive is $1/n^c$, so by a union
  bound this never occurs.  Because false negatives never occur by
  assumption, this means we continue following the ideal ordering.
  Because bad updates never occur, new residuals have magnitude at
  most
  \[
  \norm{1}{\yres} + \sqrt{\frac{\log \log n}{\log n}}\eps L.
  \]
  Because $\norm{1}{\yres}/\left(\sqrt{\frac{\log \log n}{\log n}}\eps L\right) \leq \sum_{e \in S} \text{rank}(e) = \text{rank}(v) = \text{rank}(e)-1$, each
  new residual has magnitude at most
  \begin{align}
  \label{eq:res}
  \text{rank}(e)\sqrt{\frac{\log \log n}{\log n}}\eps L \le \eps L.
  \end{align}
  as needed to complete the induction.

  Given that we follow the ideal ordering, we recover every edge of
  rank at most $\log \log n$.  Furthermore, the residue on every edge
  we recover is at most $\eps L$.  By Lemma~\ref{l:ranks} there are at
  most $k/\log^c n$ edges that we do not recover.
  From Equation~(\ref{eq:res}), the squared $\ell_2$
  norm of the residues is
   at most $\eps^2 L^2 k = \eps^2 C^2 \sigma^2 n/k \cdot k 
  < \sigma^2 n$ for $\eps$ small enough.  Since $\norm{2}{\wh{w}}^2 < 2\sigma^2 n$ with
  overwhelming probability, there exists a $\wh{z}'$ so that
  \[
  \norm{2}{\wh{z}-\wh{x} - \wh{z}'}^2 \leq 2\norm{2}{\wh{z}-\wh{x^*} -
    \wh{z}'}^2 + 2\norm{2}{w}^2 \leq 6\sigma^2 n.
  \]

  Finally, we need to bound the number of times the oracle catches
  false positives or bad updates.  The algorithm applies
  Lemma~\ref{L:FAILUREPROB} only $2\sqrt{n} + O(k) = O(k)$ times.
  Each time has a $1/\log^c n$ chance of a false positive or bad
  update.  Hence the expected number of false positives or bad updates
  is $O(k/\log^c n)$.
\end{proof}
\begin{lemma}\label{l:robust2}
  For any constant $\alpha > 0$, the algorithm \textsc{Robust2DSFFT}
  can with probability $1-\alpha$ recover $\wh{z}$ such that there
  exists a $(k/\log^{c-1} n)$-sparse $\wh{z}'$ satisfying
  \[
  \norm{2}{\wh{z}-\wh{x} - \wh{z}'}^2 \leq 6\sigma^2 n
  \]
  using $O(k \log n)$ samples and $O(k \log^2 n)$ time.
\end{lemma}
\begin{proof}
  To do this, we will show that changing the effect of a single call
  to \textsc{RobustEstimateCol} can only affect $\log n$ positions in
  the output of \textsc{Robust2DSFFT}.  By Lemma~\ref{l:robust} we can, with large constant probability
  turn \textsc{Robust2DSFFT} into \textsc{Robust2DSFFT'}
  with only $O(k/\log^{c} n)$ changes to calls to
  \textsc{RobustEstimateCol}. This means the output of
  \textsc{Robust2DSFFT} and of \textsc{Robust2DSFFT'} only
  differ in $O(k/\log^{c-1} n)$ positions.  

  We view \textsc{RobustEstimateCol} as trying to estimate a vertex.
  Modifying it can change from recovering one edge (or none) to
  recovering a different edge (or none).  Thus, a change can only
  affect at most two calls to \textsc{RobustEstimateCol} in the next
  stage.  Hence in $r$ stages, at most $2^{r-1}$ calls may be
  affected, so at most $2^r$ edges may be recovered differently.

  Because we refuse to recover any edge with rank at least $\log \log
  n$, the algorithm has at most $\log \log n$ stages.  Hence at most $\log
  n$ edges may be recovered differently as a result of a single change
  to \textsc{RobustEstimateCol}.
\end{proof}

\begin{theorem}
\label{thm:general}
  Our overall algorithm can recover $\wh{x}'$ satisfying
  \[
  \norm{2}{\wh{x} - \wh{x}'}^2 \leq 12 \sigma^2 n + \norm{2}{\wh{x}}^2/n^c
  \]
  with probability $1-\alpha$ for any constants $c, \alpha > 0$ in $O(k
  \log n)$ samples and $O(k \log^2 n)$ time, where $k=a\sqrt{n}$ for some constant $a>0$.
\end{theorem}

\begin{proof}
  By Lemma~\ref{l:robust2}, we can recover an $O(k)$-sparse $\wh{z}$
  such that there exists an $(k/\log^{c-1} n)$-sparse $\wh{z}'$ with
  \[
  \norm{2}{\wh{x} - \wh{z} - \wh{z}'}^2 \leq 6\sigma^2 n.
  \]
  with arbitrarily large constant probability for any constant $c$
  using $O(k \log^2 n)$ time and $O(k \log n)$ samples.  Then by
  Theorem~\ref{thm:2dhikp} in Appendix~\ref{app:inefficient_2D}, we
  can recover a $\wh{z}'$ in $O(k \log^2 n)$ time and $O(k
  \log^{4-c} n)$ samples satisfying
  \[
  \norm{2}{\wh{x} - \wh{z}- \wh{z}'}^2 \leq 12 \sigma^2 n +
  \norm{2}{\wh{x}}^2/n^c
  \]
  and hence $\wh{x}' := \wh{z}+\wh{z}'$ is a good reconstruction for $\wh{x}$.
\end{proof}

\bibliographystyle{alpha}
\bibliography{paper}

\appendix
\section{Sample lower bound for our distribution}\label{app:lower}

We will show that the lower bound on $\ell_2/\ell_2$ recovery
from~\cite{PW} applies to our setting with a simple reduction.  First, we state their bound:

\begin{lemma}[\cite{PW} section 4]\label{l:PW}
  For any $k < n/\log n$ and constant $\eps > 0$, there exists a
  distribution $D_k$ over $k$-sparse vectors in $\{0, 1, -1\}^n$ such
  that, for every distribution of matrices $A \in \R^{m \times n}$
  with $m = o(k \log (n/k))$ and recovery algorithms $\mathcal{A}$,
  \begin{align*}
    \Pr[\norm{2}{\mathcal{A}(A(x + w)) - x} < \sqrt{k}/5] < 1/2
  \end{align*}
  as a distribution over $x \sim D_k$ and $w \sim N(0, \sigma^2I_n)$
  with $\sigma^2 = \eps k/n$, as well as over $A$ and $\mathcal{A}$.
\end{lemma}

First, we note that we can replace $D_k$ with $U_k$, the uniform
distribution over $k$-sparse vectors in $\{0, 1, -1\}^n$ in
Lemma~\ref{l:PW}.  To see this, suppose we have an $(A, \mathcal{A})$
that works with $1/2$ probability over $U_k$.  Then for any $k$-sparse
$x \in \{0, 1, -1\}^n$, if we choose a random permutation matrix $P$
and sign flip matrix $S$, $PSx \sim U_k$.  Hence, the distribution of
matrices $APS$ and algorithm $\mathcal{A}'(x) =
\mathcal{A}((PS)^{-1}x)$ works with $1/2$ probability for any $x$, and
therefore on average over $D_k$.  This implies that $A$ has $\Omega(k
\log(n/k))$ rows by Lemma~\ref{l:PW}.  Hence, we can set $D_k = U_k$ in
Lemma~\ref{l:PW}.

Our algorithm works with $3/4$ probability over vectors $x$ that are
not necessarily $k$-sparse, but have a binomial number $B(n, k/n)$ of
nonzeros.  That is, it works over the distribution $U$ that is
$U_{k'} : k' \sim B(n, k/n)$.  With $1-e^{-\Omega(k)} > 3/4$ probability,
$k' \in [k/2, 2k]$.  Hence, our algorithm works with at least $1/2$
probability over $(U_{k'} : k' \sim B(n, k/n) \cap k' \in [k/2, 2k])$.
By an averaging argument, there must exist a $k' \in [k/2, 2k]$ where
our algorithm works with at least $1/2$ probability over $U_{k'}$; but
the lemma implies that it must therefore take $\Omega(k' \log (n/k'))
= \Omega(k \log (n/k))$ samples.

\section{Robust 2D FFTs}
\label{app:inefficient_2D}

This section outlines the straightforward generalization
of~\cite{HIKP2} to two dimensions, as well as how to incorporate the extra parameter $\wh{z}$ of already recovered coefficients.   Relative to our result of
Theorem~\ref{thm:general}, this result takes  more samples.
However, it does not require that the input be from a random
distribution and is used as a subroutine by Theorem~\ref{thm:general}
after decreasing the sparsity by a $\log^c n$ factor.

Because we use this as a subroutine after computing an estimate
$\wh{z}$ of $\wh{x}$, we actually want to estimate $\wh{x} - \wh{z}$
where we have oracle access to $x$ and to $\wh{z}$.

\begin{theorem}\label{thm:2dhikp}
  There is a variant of \cite{HIKP2} algorithm that will, given $x, \wh{z} \in
  \C^{\sqrt{n} \times \sqrt{n}}$, return $\wh{x'}$ with
  \[
  \norm{2}{\wh{x} - \wh{z} - \wh{x'}} \leq 2 \cdot
  \min_{k\text{-sparse } \wh{x^*}} \norm{2}{\wh{x} - \wh{z} -
    \wh{x^*}}^2 + \norm{2}{\wh{x}}^2 / n^c
  \]
  with probability $1-\alpha$ for any constants $c, \alpha > 0$ in
  time
  \[O( k \log (n/k)\log^2 n + \abs{\supp(\wh{z})} \log (n/k)\log
  n),\] using $O( k \log (n/k)\log^2 n)$ samples of $x$.\xxx{Added
    extra log}
\end{theorem}

\begin{proof}
  We need to modify \cite{HIKP2} in two ways: by extending it to two dimensions and
  by allowing the parameter $\wh{z}$.  We will start by describing the adaptation to two dimensions. 

  The basic idea of~\cite{HIKP2} is to construct from Fourier
  measurements a way to ``hash'' the coordinates in $B = O(k)$ bins.
  There are three basic components that are needed: a
  \emph{permutation} that gives nearly pairwise independent hashing to
  bins; a \emph{filter} that allows for computing the sum of bins
  using Fourier measurements; and the \emph{location} estimation needs
  to search in both axes.  The permutation is the main subtlety.

  \paragraph{Permutation}  Let $\mathcal{M} \subset [\sqrt{n}]^{2
    \times 2}$ be the set of matrices with odd determinant.  For
  notational purposes, for $v = (i, j)$ we define $x_v := x_{i,j}$.

  \begin{definition}
    For $M \in \mathcal{M}$ and $a, b \in [\sqrt{n}]^2$ we define the
    \emph{permutation} $P_{M, a, b} \C^{\sqrt{n}\times\sqrt{n}} \to
    \C^{\sqrt{n}\times\sqrt{n}}$ by
    \[
    (P_{M, a, b}x)_v  = x_{M (v-a)}\omega^{v^TMb}.
    \]
    We also define $\pi_{M, b}(v) = M(v - b) \mod \sqrt{n}$.
  \end{definition}
  \begin{claim}
    $\wh{P_{M,a,b}x}_{\pi_{M^T,b}(v)} = \wh{x}_v\omega^{v^TM^Ta}$
  \end{claim}
  \begin{proof}
      \begin{align*}
    \wh{P_{M,a,b}x}_{M(v - b)} 
    &= \frac{1}{\sqrt{n}} \sum_{u \in [\sqrt{n}]^2} \omega^{u^TM(v-b)}(P_{M,a,b}x)_u\\
    &= \frac{1}{\sqrt{n}} \sum_{u \in [\sqrt{n}]^2} \omega^{u^TM(v-b)}x_{M (u-a)}\omega^{u^TMb}\\
    &= \omega^{v^TM^Ta}\frac{1}{\sqrt{n}} \sum_{u \in [\sqrt{n}]^2} \omega^{v^TM^T (u-a)}x_{M (u-a)}\\
    &= \wh{x}_i\omega^{v^TM^Ta}
  \end{align*}
  where we used that $M^T$ is a bijection over $[\sqrt{n}]^2$ because
  $\det(M)$ is odd.
  \end{proof}

  This gives a lemma analogous to Lemma~2.4 of~\cite{HIKP2}.

  \begin{lemma}\label{l:nearlyindependent}
    Suppose $v \in [\sqrt{n}]^2$ is not $0$.  Then
    \[
    \Pr_{M \sim \mathcal{M}}[ Mv \in [-C, C]^2 \pmod {\sqrt{n}}] \leq O(\frac{C^2}{n}).
    \]
  \end{lemma}
  \begin{proof}
    For any $u$, define $G(u)$ to be the largest power of $2$ that
    divides both $u_0$ and $u_1$. Define $g = G(v)$, and let $S = \{u
    \in [\sqrt{n}]^2 \mid G(u) = g\}$.  We have that $Mv$ is uniform
    over $S$: $\mathcal{M}$ is a group and $S$ is the orbit of $(0,
    g)$.

    Because $S$ lies on a lattice of distance $g$ and does not include
    the origin, there are at most $(2\floor{C/g}+1)^2-1 \leq 8(C/g)^2$
    elements in $S \cap [-C, C]^2$, and $(3/4)n/g^2$ total elements in
    $S$.  Hence the probability is at most $(32/3) C^2/n$.
  \end{proof}

  We can then define the ``hash function'' $h_{M,b} : [\sqrt{n}]^2 \to
  [\sqrt{B}]^2$ given by $(h_{M,b}(u)) = \round(\pi_{M,b}(u) \cdot
  \sqrt{n/B})$; i.e., round to the nearest multiple of $\sqrt{n/B}$ in
  each coordinate and scale down.  We also define the ``offset''
  $o_{M,b}(u) = \pi_{M,b}(u) - \sqrt{n/B}h_{M,b}(u)$.  This lets us
  give results analogous to Claims 3.1 and 3.2 of~\cite{HIKP2}:
  \begin{itemize}
  \item $\Pr[h_{M,b}(u) = h_{M,b}(v) < O(1/B)]$ for $u \neq v$.  In order for
    $h(u) = h(v)$, we need that $\pi_{M,b}(u) - \pi_{M,b}(v) \in
    [-2\sqrt{n/B}, 2\sqrt{n/B}]^2$.  But Lemma~\ref{l:nearlyindependent}
    implies this probability is $O(1/B)$.
  \item $\Pr[o_{M,b}(u) \notin [-(1-\alpha)\sqrt{n/B},
    (1-\alpha)\sqrt{n/B}]^2] < O(\alpha)$ for any $\alpha > 0$.
    Because of the offset $b$, $o_{M,b}(u)$ is uniform over
    $[-\sqrt{n/B}, \sqrt{n/B}]^2$.  Hence the probability is $2\alpha
    - \alpha^2 + o(1)$ by a volume argument.
  \end{itemize}
  which are all we need of the hash function.

  \paragraph{Filter}  Modifying the filter is pretty simple.
  Specifically,\cite{HIKP2} defined a filter $G \in \R^{\sqrt{n}}$ with support
  size $O(\sqrt{B} \log n)$ such that $\wh{G}$ is essentially zero
  outsize $[-\sqrt{n/B}, \sqrt{n/B}]$ and is essentially $1$ inside
  $[-(1-\alpha)\sqrt{n/B}, (1-\alpha)\sqrt{n/B}]$ for constant
  $\alpha$.  We compute the $\sqrt{B} \times \sqrt{B}$ 2-dimensional
  DFT of $x'_{i,j} = x_{i,j}G_i G_j$ to sum up the element in each
  bin.  This takes $B \log^2 n$ samples and time rather than $B \log
  n$, which is the reason for the extra $\log n$ factor compared to the one dimensional
  case.

  \paragraph{Location}  Location is easy to modify; we simply run it
  twice to find the row and column separately.
  
  In summary, the aforementioned adaptations leads to a variant of the ~\cite{HIKP2} algorithm that works in two dimensions, with running time $O( k \log (n/k)\log^2 n)$, using $O( k
  \log (n/k)\log^2 n)$ samples.
  
  \paragraph{Adding extra coefficient list  $\wh{z}$}
   The modification of the algorithm of~\cite{HIKP2} (as well as its variant above) is straightforward. The algorithm performs a sequence of iterations, where each iteration involves hashing the frequencies of the signal into bins, followed by subtracting the already recovered coefficients from the bins. Since the algorithm recovers $\Theta(k)$ coefficients in the first iteration, the subtracted list is always of size 
   $\Theta(k)$. 
   
   Given the extra coefficient list, the only modification to the
   algorithm is that the list of the subtracted coefficients needs to
   be appended with coefficients in $\wh{z}$.  Since this step does
   not affect the samples taken by the algorithm, the sample bound
   remains unchanged. To analyze the running time, let $k'$ be the
   number of nonzero coefficients in $\wh{z}$. Observe that the total
   time of the original algorithm spent on subtracting the
   coefficients from a list of size $\Theta(k)$ was $O( k \log
   (n/k)\log n)$, or $O(\log (n/k)\log n)$ per list coefficient. Since
   in our case the number of coefficients in the list is increased
   from $\Theta(k)$ to $k'+\Theta(k)$, the running time is increased
   by an additive factor of $O(k' \log (n/k)\log n)$.
\end{proof}

\end{document}